\documentclass[a4paper]{article}

\usepackage{geometry}

\usepackage{palatino}
\usepackage{latexsym}

\usepackage[mathscr]{eucal}
\usepackage{amsmath}
\usepackage{amsthm}
\usepackage{amsfonts}
\usepackage{amssymb}
\usepackage{amscd}
\usepackage{color}
\usepackage{graphicx}
\usepackage{graphics}
\usepackage{pifont}
\usepackage{subfigure}
\usepackage[makeroom]{cancel}
\usepackage[normalem]{ulem}
\usepackage[dvipsnames]{xcolor}
\usepackage{tikz}

\theoremstyle{plain}
\newtheorem{theorem}{Theorem}
\newtheorem{remark}[theorem]{Remark}
\newtheorem{lemma}[theorem]{Lemma}
\newtheorem{proposition}[theorem]{Proposition}
\newtheorem{definition}[theorem]{Definition}

\usepackage{newtxtext}       %
\usepackage{newtxmath}  
\usepackage{tikz}

\newcommand{\ii}{\mathrm{i}}

\title{Classical and quantum controllability of a rotating symmetric molecule}

\begin{document}
\author{
Ugo Boscain\footnote{CNRS, Sorbonne Universit\'{e}, Inria, Universit\'{e} de Paris, Laboratoire Jacques-Louis Lions, Paris, France.
              (ugo.boscain@upmc.fr)  }        \and
         Eugenio~Pozzoli\footnote{Inria, Sorbonne Universit\'e, Universit\'e de Paris, CNRS, Laboratoire Jacques-Louis Lions, Paris, France(eugenio.pozzoli@inria.fr).}      \and
    Mario Sigalotti\footnote{Inria, Sorbonne Universit\'e, Universit\'e de Paris, CNRS, Laboratoire Jacques-Louis Lions, Paris, France (mario.sigalotti@inria.fr).     }}
         
         \maketitle
\begin{abstract}
 In this paper we study the controllability problem for a symmetric-top molecule, both for its classical and quantum rotational dynamics. 
The molecule is controlled through
 three orthogonal 
 electric fields 
 interacting with its electric dipole.
 We characterize the controllability in terms of the dipole position: when it lies along the symmetry axis of the molecule neither the classical nor the quantum dynamics are controllable, due to the presence of a conserved quantity, the third component of the total angular momentum; when it lies in the orthogonal plane to the symmetry axis, a quantum symmetry arises, due to the superposition of symmetric states, which has no classical counterpart. If the dipole is neither along the symmetry axis nor orthogonal to it, controllability for the classical dynamics and approximate controllability for the quantum dynamics are proved to hold. 
 The approximate controllability of the symmetric-top Schr\"odinger equation is established by  using a Lie--Galerkin method, based on block-wise approximations of the infinite-dimensional systems.
\end{abstract}
\textbf{Keywords:} Quantum control, Schr\"odinger equation, rotational dynamics, symmetric-top molecule, bilinear control systems, Euler equation

\section{Introduction}
The control of molecular dynamics takes an important role in quantum physics and chemistry because of the variety of its applications, 
starting from well-established ones such as rotational state-selective excitation of chiral molecules (\cite{perez,sandra}), and going further to applications in quantum information (\cite{yu}). For a general overview of controlled molecular dynamics one can see, for example, \cite{koch}.  

Rotations can, in general, couple to vibrations in the so-called ro-vibrational states. In our mathematical analysis, however, we shall restrict ourselves to the rotational states of the molecule. Due to its discrete quantization, molecular dynamics 
perfectly fits the mathematical quantum control theory which has been established until now. In fact, the control of 
the Schr{\"o}dinger equation has attracted substantial interest in the last 15 years
(see 
\cite{Altafini,Coron,
BGRS,
Glaser2015,Keyl,nersesyan} and references therein). 
Rigid molecules are subject to the classification of rigid rotors in terms of their inertia moments $I_1\leq I_2\leq I_3$: one distinghuishes asymmetric-tops ($I_1<I_2<I_3$), prolate symmetric-tops ($I_1<I_2=I_3$), oblate symmetric-tops ($I_1=I_2<I_3$), spherical-tops ($I_1=I_2=I_3$), and linear-tops ($I_1=0,\,I_2=I_3$).

The problem of controlling the rotational dynamics of a planar molecule 
 by means of two orthogonal 
 electric fields
has been analyzed in \cite{BCCS}, where 
approximate controllability has been proved using 
a suitable non-resonance property of
the spectrum of the rotational Hamiltonian. 
In \cite{BCS} the approximate controllability of a linear-top controlled by three orthogonal 
electric fields has been established. There, a new sufficient condition for controllability, called the Lie--Galerkin tracking condition, has been introduced in an abstract framework, and applied to the linear-top system.

Here, we study the symmetric-top (prolate, oblate, or spherical)  as a generalization of the linear one, characterizing its controllability in terms of the position of its electric dipole moment.
While for the linear-top two quantum numbers $j,m$ are needed to describe the motion, the main and more evident difference here is the presence of a third quantum number $k$, which classically represents the projection of the total angular momentum on the symmetry axis of the molecule. This should not be a surprise, since the configuration space of a linear-top is the $2$-sphere $S^2$, while the symmetric-top evolves on the Lie group ${\rm SO}(3)$, a three-dimensional manifold. As a matter of fact, by fixing $k=0$, one recovers the linear-top as a subsystem inside the symmetric-top. It is worth mentioning that the general theory developed in \cite{BCCS,BCMS,nersesyan} is based on non-resonance conditions on the spectrum of the internal Hamiltonian. A major difficulty in studying the controllability properties of the rotational dynamics is that, even in the case of the linear-top, the spectrum of the rotational Hamiltonian has severe degeneracies at the so-called $m$-levels. The symmetric-top is even more degenerate, due to the additional presence of the so-called $k$-levels. 

The Schr{\"o}dinger equation for a rotating molecule controlled by three orthogonal
 electric fields reads
\[
\ii\dfrac{\partial}{\partial t} \psi(R,t)= H\psi(R,t)+\sum_{l=1}^3u_l(t)B_l(R,\delta)\psi(R,t), \quad \psi(\cdot,t) \in L^2({\rm SO}(3)),
\]
where $H=\frac{1}{2}\Big(\frac{P_1^2}{I_1}+\frac{P_2^2}{I_2}+\frac{P_3^2}{I_3}\Big)$ is the rotational Hamiltonian, $I_1,I_2,I_3$ are the moments of inertia of the molecule, $P_1,P_2,P_3$ are the angular momentum differential operators, and $B_i(R,\delta)=-\langle R \delta, e_i\rangle$ is the interaction Hamiltonian between the dipole moment $\delta$ of the molecule and the direction $e_i$, $i=1,2,3$. Finally, $R \in {\rm SO}(3)$ is the matrix which describes the configuration of the molecule in the space.

\begin{figure}[ht!]
\subfigure[]{
\includegraphics[width=0.3\linewidth, draft = false]{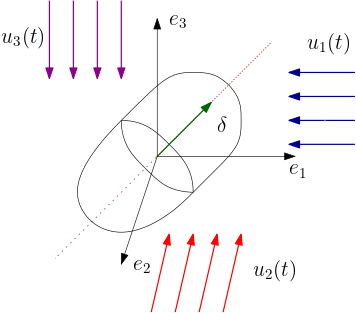} \label{top1} }
\subfigure[]{
\includegraphics[width=0.3\linewidth, draft = false]{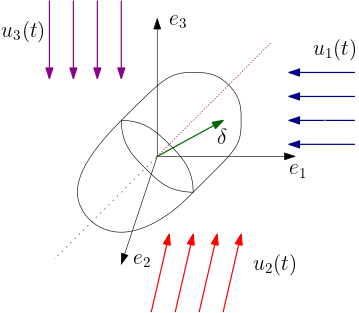} \label{top2} }
\subfigure[]{
\includegraphics[width=0.3\linewidth, draft = false]{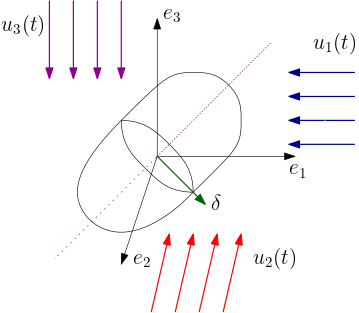} \label{top3} }
\label{top}
\end{figure}

We shall study the symmetric-top, and set $I_1=I_2$. Anyway, our analysis does not depend on whether $I_3\geq I_2$ or $I_3\leq I_2$, so we are actually treating in this way both the cases of a prolate or oblate symmetric-top.
The principal axis of inertia with associated inertia moment $I_3$ is then called \emph{symmetry axis} of the molecule. The position of the electric dipole with respect to the symmetry axis plays a crucial role in our controllability analysis: a symmetric molecule with electric dipole collinear to the symmetry axis will be called \emph{genuine}, otherwise it will be called \emph{accidental} (\cite[Section 2.6]{gordy}). 
Most symmetric molecules present in nature are genuine. Nevertheless, it can happen that two moments of inertia of a real molecule are almost equal, by ``accident", although the molecule does not possess a $n$-fold axis of symmetry with $n\geq 3$\footnote{ The existence of a $n$-fold axis of symmetry (i.e., an axis such that a rotation of angle $2\pi/n$ about it leaves unchanged the distribution of atoms in the space) with $n\geq 3$, implies that the top is genuine symmetric.}For instance,  the inertia moments of the 
chiral molecule HSOH are $I_1\sim I_2\ll I_3$, 
while its dipole components are $\delta_1>\delta_2=\delta_1/2\gg \delta_3\neq 0$ (\cite{WinnewisserCEJ2003}). Such slightly asymmetric-tops 
are often studied in chemistry and physics in their symmetric-top approximations (see, e.g., \cite{WinnewisserCEJ2003},\cite[Section 3.4]{gordy}), which correspond in general to accidentally symmetric-tops. In this case, closed expression for the spectrum and the eigenfunctions of $H$ are known. The case of the asymmetric-top goes beyond the scope of this paper, but we remark  that accidentally symmetric-tops may be used to obtain controllability of asymmetric-tops with a perturbative approach.
The idea of studying the controllability of quantum systems in general configurations starting from symmetric cases (even if the latter have more degeneracies) has already been exploited, e.g., in \cite{panati,mehats}.

The position of the dipole moment turns out to play a decisive role: when it is neither along the symmetry axis, nor orthogonal to it, as in 
Figure~\ref{top}\subref{top2}, then approximate controllability holds, under some non-resonance conditions, as it is stated in Theorem~\ref{rare}. To prove it, we introduce in Theorem \ref{LGTC} 
a new controllability test for the multi-input Schr{\"o}dinger equation, 
closely related to the Lie--Galerkin tracking condition.
We then apply this result to the symmetric-top system.
 The control strategy is based on the excitation of the system with external fields in resonance with three families of frequencies, corresponding to internal spectral gaps. 
One frequency is used to overcome the $m$-degeneracy in the spectrum, and this step is quite similar to the proof of the linear-top approximate controllability (Appendix~\ref{appendixA}). The other two frequencies are used in a next step to break the $k$-degeneracy, in a three-wave mixing scheme (Appendix~\ref{appendixB}) typically used in quantum chemistry to obtain enantio- and state-selectivity for chiral molecules (\cite{AGGT,GKL,YY}).

The two dipole configurations to which Theorem \ref{rare} does not apply are extremely relevant from the physical point of view. Indeed, the dipole moment of a symmetric-top lies usually along its symmetry axis (Figure~\ref{top}\subref{top1}), and if not, for accidentally symmetric-tops, it is often found in the orthogonal plane (Figure~\ref{top}\subref{top3}). Here two different symmetries arise, implying the non-controllability of these systems, as we prove, respectively, in Theorems~\ref{genuine} and \ref{accidentally}.
These two conserved quantities stimulated and motivated the study of the classical dynamics of the symmetric-top, presented in the first part of the paper: the first conserved quantity, appearing in Theorem~\ref{genuine}, corresponds to a classical observable, that is, the component of the angular momentum along the symmetry axis, and it turns out to be a first integral also for the classical controlled dynamics, as remarked in Theorem~\ref{genuinecla}. 
The second conserved quantity, appearing in Theorem~\ref{accidentally}, is more challenging, because it does not have a counterpart in the classical dynamics, being mainly due to the superposition of $k$ and $-k$ states in the quantum dynamics. We show that this position of the dipole still corresponds to a controllable system for the classical-top, while it does not for the quantum-top. Thus, the latter is an example of a system whose quantum dynamics are not controllable even though the classical dynamics are. The possible discrepancy between quantum and classical controllability has been already noticed, for example, in the harmonic oscillator dynamics (\cite{Mirra}).
 It should be noticed that the classical dynamics of a rigid body controlled with external torques (e.g., opposite pairs of gas jets) or internal torques (momentum exchange devices such as wheels) as studied in the literature (see, e.g., \cite[Section 6.4]{AS},
  \cite{krishna},\cite{crouch}, \cite[Section 4.6]{Jurdje}) differ from the ones considered here, where  the controlled fields (i.e., the interaction between 
 the 
 electric field and the electric dipole) are not left-invariant and their action depends on the configuration of the rigid body in the space.

The paper is organized as follows: in Section \ref{classical} we study the controllability of the classical Hamilton equations for a symmetric-top. The main results are 
Theorems \ref{genuinecla} and \ref{accidentallycla}, where we prove, respectively, the non-controllability when the dipole lies along the symmetry axis of the body and the controllability in any other case. In Section \ref{quantum} we study the controllability of the Schr{\"o}dinger equation for a symmetric-top. The main controllability result is Theorem \ref{rare}, where we prove the approximate controllability when the dipole is neither along the symmetry axis, nor orthogonal to it. In the two cases left, we prove the non-controllability in Theorems~\ref{genuine} and \ref{accidentally}.

\section{Classical symmetric-top molecule}\label{classical}

\subsection{Controllability of control-affine systems with recurrent drift}
We recall in this section some useful results on the controllability properties of (finite-dimensional) control-affine systems.

Let $M$ be an $n$-dimensional manifold, $X_0,X_1,\dots,X_\ell$ a family of smooth (i.e., $C^\infty$) vector fields on $M$, 
$U\subset \mathbb{R}^\ell$ a set of control values which is 
a neighborhood of the origin. 
We consider the control system
\begin{equation}\label{control}
\dot{q}=X_0(q)+\sum_{i=1}^\ell u_i(t)X_i(q), \qquad q \in M,
\end{equation} 
where the control functions $u$ are taken in $L^\infty(\mathbb{R},U)$. The vector field $X_0$ is called the \emph{drift}.
The \emph{reachable set} from $q_0\in M$ is
\begin{align*}
 \mathrm{Reach}(q_0):= &\{q \in M \mid \exists \;  u,T \text{ s.t. the solution to (\ref{control}) with } q(0)=q_0 \\ & \text{ satisfies } q(T)=q \}. 
 \end{align*}
System (\ref{control}) is said to be \emph{controllable} if $\mathrm{Reach}(q_0)=M$ for all $q_0\in M$.

The family of vector fields $X_0,X_1,\dots,X_\ell$ is said to be \emph{Lie bracket generating} if 
$$\dim(\mathrm{Lie}_q\{X_0,X_1,\dots,X_\ell \})=n $$ 
for all $q\in M$, where $\mathrm{Lie}_q\{X_0,X_1,\dots,X_\ell \}$ denotes the evaluation at $q$ of the Lie algebra generated by $X_0,X_1,\dots,X_\ell$.

The following is a basic result in geometric control theory (see, for example, \cite[Section 4.6]{Jurdje}). Recall that a complete vector field $X$ on $M$
is said to be \emph{recurrent} if for every open nonempty subset $V$ of $M$ and every time $t>0$, there exists $\tilde{t}>t$ such that $\phi_{\tilde{t}}(V)\cap V \neq \emptyset$, where $\phi_{\tilde{t}}$ denotes the flow of $X$ at time ${\tilde{t}}$.

\begin{theorem}\label{basic}
Let $U\subset \mathbb{R}^m$ be 
a neighborhood of the origin. 
If $X_0$ is recurrent and the family $X_0,X_1,\dots,X_{\ell}$ is Lie bracket generating, then system (\ref{control}) is controllable.
\end{theorem}

A useful test to check that the Lie bracket generating condition holds true
is given by the following simple lemma, whose proof is given for completeness.
\begin{lemma}\label{lemmino}
If the family of analytic vector fields $X_0,X_1,\dots,X_{\ell}$ is Lie bracket generating on the complement of a subset
$N\subset M$ and $\mathrm{Reach}(q)\not\subset N$, for all $q\in N$, then the family is Lie bracket generating on $M$.
\end{lemma}
\begin{proof}
Let $q\in N$ and $q_1\in \mathrm{Reach}(q)\setminus N$.
By the Orbit theorem applied to the case of analytic vector fields (see, e.g., \cite[Chapter 5]{AS}) the dimension of 
$\mathrm{Lie}_q\{X_0,X_1,\dots,X_{\ell} \}$ and $\mathrm{Lie}_{q_1}\{X_0,X_1,\dots,X_{\ell} \}$ coincide. By assumption the latter is equal to $n$, which implies that the same is true for the former.
\end{proof}

\subsection{The classical dynamics of a molecule subject to electric fields}
Since the translational motion (of the center of mass) of a rigid body is decoupled from the rotational motion, we shall assume that the molecule can only rotate around its center of mass. In detail, for any vector $v\in \mathbb{R}^3$, denoting by $e_1,e_2,e_3$ a fixed orthonormal frame of $\mathbb{R}^3$ and by $a_1,a_2,a_3$ a moving orthonormal frame with the same orientation, both attached to the rigid body's center of mass, the configuration of the molecule is identified with the unique $g \in {\rm SO}(3)$ such that $g\;(x,y,z)^T=(X,Y,Z)^T$, where $(x,y,z)$ are the coordinates of $v$ with respect to $a_1,a_2,a_3$, and $(X,Y,Z)$ are the coordinates of $v$ with respect to $e_1,e_2,e_3$. In order to describe the equations on the tangent bundle ${\rm SO}(3)\times \mathfrak{so}(3)$, 
we shall make use of the isomorphism of Lie algebras
\[
A:(\mathbb{R}^3,\times) \rightarrow (\mathfrak{so(3)},[\cdot,\cdot]), \quad P=
\begin{pmatrix}
P_1 \\
P_2\\
P_3
\end{pmatrix} \mapsto  A(P)=
\begin{pmatrix}
0 & -P_3 & P_2\\
P_3 & 0 & -P_1\\
-P_2 & P_1 & 0
\end{pmatrix}
\]
where $\times$ is the vector product. As external forces to control the rotation of the molecule, we consider three orthogonal
electric fields with intensities $u_1(t)$, $u_2(t)$, $u_3(t)$ and directions $e_1,e_2,e_3$.
We assume that 
$$(u_1,u_2,u_3)\in U\subset \mathbb{R}^3, \quad  (0,0,0)\in {\rm Interior}(U), $$
that is, the set $U\subset \mathbb{R}^3$ of admissible values
for the triple $(u_1,u_2,u_3)$ is 
a neighborhood of the origin.
 Denoting by $\delta$ the dipole of the molecule written in the moving frame, the three forces due to the interaction with the electric fields are
$
u_i(t)(g^{-1}(t)e_i)\times \delta$, $i=1,2,3.$
Then, the equations for the classical rotational dynamics of a molecule with inertia moments $I_1,I_2,I_3$ controlled with electric fields read
\begin{equation}\label{euler1}
\begin{pmatrix}
\dot{g} \\ 
\dot{P}
\end{pmatrix}=X(g,P)+\sum_{i=1}^3u_i(t)Y_i(g,P), \quad (g,P)\in {\rm SO}(3) \times \mathbb{R}^3,\ u\in U,
\end{equation}
where
 \begin{equation}\label{fields}
X(g,P):=\begin{pmatrix}
gA(\beta P) \\
P \times (\beta P)
\end{pmatrix}, \quad Y_i(g,P):=\begin{pmatrix}
0\\
(g^{-1}e_i) \times \delta
\end{pmatrix}, \quad i=1,2,3,
\end{equation}
and $P=(P_1,P_2,P_3)^T,\;\beta P=(P_1/I_1,P_2/I_2,P_3/I_3)^T$. Similarly to
\cite[Section 12.2]{Jurdje} (where this is done for the heavy rigid body), these equations can be derived as Hamilton equations corresponding to the Hamiltonian
$$H=\frac{1}{2}\left(\frac{P_1^2}{I_1}+\frac{P_2^2}{I_2}+\frac{P_3^2}{I_3}\right)+V(g),\quad  V(g)=-\sum_{i=1}^3u_i\langle(g^{-1}e_i), \delta\rangle $$
on ${\rm SO}(3)\times \mathbb{R}^3$.
System \eqref{euler1} can be seen as a control-affine system with $\ell=3$ 
controlled fields. 

Rotating molecule dynamics can also be represented in terms of quaternions, lifting the dynamics from  ${\rm SO}(3)$ to  the $3$-sphere $S^3$, as follows.
We 
denote by $\mathbb{H}$ the space of quaternions and we
identify $S^3\subset \mathbb{R}^4$ with $\{q_0+\ii q_1+\mathrm{j}q_2+\mathrm{k}q_3\in \mathbb{H} \mid q_0^2+q_1^2+q_2^2+q_3^2=1\}$. We also identify $\mathbb{R}^3$ with $\{\ii P_1+\mathrm{j}P_2+\mathrm{k}P_3 \in \mathbb{H}\mid (P_1,P_2,P_3)\in \mathbb{R}^3\}$.
Via this identification, the vector product $P\times \Omega$ becomes $\frac{1}{2}[P,\Omega]:=\frac{1}{2}(P\Omega-\Omega P)$, for any $P,\Omega \in \mathbb{R}^3$.
Moreover, given $q=\cos(\alpha)+(q_1,q_2,q_3)\sin(\alpha) \in S^3$ and $P \in \mathbb{R}^3$, 
the quaternion product 
 $qP\overline{q}$ is in $\mathbb{R}^3$ and corresponds to the rotation
 of $P$ of angle $2\alpha$ around the axis $(q_1,q_2,q_3)$. 
Hence, $S^3$ can be seen as a double covering space of ${\rm SO}(3)$ (see \cite[Section 5.2]{Ratiu} for further details).
System (\ref{euler1}) is lifted to $S^3\times \mathbb{R}^3$ to the system
\begin{equation}\label{quaternion}
\begin{cases}
\begin{aligned}
\dfrac{dq(t)}{dt}=&q(t)\beta P(t), \\ 
\dfrac{dP(t)}{dt}=&\frac{1}{2}[P(t), \beta P(t)]+\dfrac{u_1(t)}{2}[\overline{q(t)}\ii q(t),\delta]+\frac{u_2(t)}{2}[\overline{q(t)}\mathrm{j}q(t),\delta]\\
&+\frac{u_3(t)}{2}[\overline{q(t)}\mathrm{k}q(t),\delta].
\end{aligned}
\end{cases}
\end{equation}
We are going to use the quaternion representation in order to prove that 
the vector fields characterizing 
 \eqref{quaternion} form a Lie bracket generating family. As a consequence, 
 the same will be true for \eqref{euler1}.

\subsection{Non-controllability of the classical genuine symmetric-top}

In most cases of physical interest, the electric dipole $\delta$ of a symmetric-top molecule lies along the symmetry axis of the molecule. If $I_1=I_2$, the symmetry axis is the third one, and we have that $\delta=
(0,0,\delta_3)^T$, $\delta_3 \neq 0$, in the body frame. The corresponding molecule is called a \emph{genuine symmetric-top} (\cite[Section 2.6]{gordy}). 

\begin{theorem}\label{genuinecla}
The third angular momentum $P_3$ is a conserved quantity for the controlled motion \eqref{euler1} of the genuine symmetry-top molecule.
\end{theorem}
\begin{proof}
In order to compute the equation satisfied by $P_3$ in \eqref{euler1}, notice that
\begin{align*}
P(t)\times \beta P(t)=\begin{pmatrix}
P_1(t) \\
P_2(t) \\
P_3(t)
\end{pmatrix}\times\begin{pmatrix}
P_1(t)/I_2 \\
P_2(t)/I_2 \\
P_3(t)/I_3
\end{pmatrix}=\begin{pmatrix}
\Big(\frac{1}{I_3}- \frac{1}{I_2}\Big)P_2(t)P_3(t) \\
\Big(\frac{1}{I_2}- \frac{1}{I_3}\Big)P_1(t)P_3(t)\\
0
\end{pmatrix}.
\end{align*}
Moreover, $u_i(t)(g^{-1}(t)e_i)\times\delta=u_i(t)(g^{-1}(t)e_i)\times (0,0,\delta_3)^T=
(\star,\star,0)^T$.
Hence, for a genuine symmetric-top, 
 the equation for $P_3$ becomes $\frac{dP_3(t)}{dt}=0$.
\end{proof}

As a consequence, the controlled dynamics live in the hypersurfaces $\{P_3=\mathrm{const}\}$ and hence system (\ref{euler1}) is not controllable in the $6$-dimensional manifold ${\rm SO}(3)\times \mathbb{R}^3$.

\subsection{Controllability of the classical accidentally symmetric-top}
In Theorem \ref{genuinecla} we proved that $P_3$ is a first integral for equations (\ref{euler1}), using both the symmetry of the mass and the symmetry of the charge, meaning that $I_1=I_2$ and $\delta=(0,0,\delta_3)^T$. 
We consider now a symmetric-top molecule with electric dipole $\delta$ not along the symmetry axis of the body, that is, 
$\delta=(\delta_1,\delta_2,\delta_3)^T$, with $\delta_1 \neq 0$ or $\delta_2 \neq 0$. This system is usually called \emph{accidentally symmetric-top} (\cite[Section 2.6]{gordy}).

\begin{theorem}\label{accidentallycla}
For an accidentally symmetric-top molecule system \eqref{euler1} is controllable.
\end{theorem}
\begin{proof}
The drift $X$ is recurrent, as observed 
in \cite[Section 8.4]{AS}. Thus, by Theorem~\ref{basic}, to prove controllability it suffices to show that, for any $(g,P) \in {\rm SO}(3) \times \mathbb{R}^3$, $\mathrm{dim}\Big( \mathrm{Lie}_{(g,P)}\{X,Y_1,Y_2,Y_3\}\Big)=6$. Actually, we will 
find six vector fields in 
$\mathrm{Lie}\{X,Y_1,Y_2,Y_3\}$ whose span is six-dimensional everywhere but on a set of positive codimension, and we will conclude by applying Lemma \ref{lemmino}. Notice that
$[X,Y_i](g,P)=\begin{pmatrix}
-gS(\beta[(g^{-1})e_i \times \delta]) \\
\star
\end{pmatrix}.$
 Denote by $\Pi_{{\rm SO}(3)}$ the projection onto the ${\rm SO}(3)$ part of the tangent bundle, that is, $\Pi_{{\rm SO}(3)}: T({\rm SO}(3)\times \mathbb{R}^3)\rightarrow T{\rm SO}(3).$
Then we have
\begin{align*}
\mathrm{span}&\{\Pi_{{\rm SO}(3)}X(g,P),\Pi_{{\rm SO}(3)}[X,Y_1](g,P),  \Pi_{{\rm SO}(3)}[X,Y_2](g, P),\Pi_{{\rm SO}(3)}[X,Y_3](g,P) \} \\ &=gS\Big(\beta[\{\delta\}^\perp \oplus \mathrm{span}\{P\}]\Big).
\end{align*}
Hence, if $\langle P, \delta \rangle \neq 0$, we have
\begin{align}\nonumber
\dim & \Big( \mathrm{span}\{\Pi_{{\rm SO}(3)}X(g,P),\Pi_{{\rm SO}(3)}[X,Y_1](g,P),\Pi_{{\rm SO}(3)}[X,Y_2](g,P),\\ &\Pi_{{\rm SO}(3)}[X,Y_3](g,P) \} \Big) 
=3 .\label{3dim}
\end{align}

To go further in the analysis, it is convenient to use the quaternion parametrization (\ref{quaternion}) in which every field is polynomial. We have, in coordinates $q=(q_0,q_1,q_2,q_3) \in S^3, P=(P_1,P_2,P_3) \in \mathbb{R}^3$,
\[
X(q,P)=\begin{pmatrix}
q\beta P \\
\frac{1}{2}[P,\beta P]
\end{pmatrix}=\begin{pmatrix}
-q_1\frac{P_1}{I_2}-q_2\frac{P_2}{I_2}-q_3\frac{P_3}{I_3}\\[1mm]
q_0\frac{P_1}{I_2}+q_2\frac{P_3}{I_3}-q_3\frac{P_2}{I_2}\\[1mm]
q_0\frac{P_2}{I_2}-q_1\frac{P_3}{I_3}+q_3\frac{P_1}{I_2}\\[1mm]
q_0\frac{P_3}{I_3}+q_1\frac{P_2}{I_2}-q_2\frac{P_1}{I_2}\\[1mm]
\Big(\frac{1}{I_3}- \frac{1}{I_2}\Big)P_2P_3 \\[1mm]
\Big(\frac{1}{I_2}- \frac{1}{I_3}\Big)P_1P_3\\[1mm]
0
\end{pmatrix}, \]
\[
Y_1(q,P)=\begin{pmatrix}
0\\
\frac{1}{2}[\overline{q}\ii q,\delta]
\end{pmatrix}=\begin{pmatrix}
0\\
0\\
0\\
0\\
(q_1q_2-q_0q_3)\delta_3-(q_1q_3+q_0q_2)\delta_2\\
(q_1q_3+q_0q_2)\delta_1-\frac{1}{2}(q_0^2+q_1^2-q_2^2-q_3^2)\delta_3\\
\frac{1}{2}(q_0^2+q_1^2-q_2^2-q_3^2)\delta_2-(q_1q_2-q_0q_3)\delta_1
\end{pmatrix},
\]
\[
 Y_2(q,P)=\begin{pmatrix}
0\\
\frac{1}{2}[\overline{q}\mathrm{j}q,\delta]
\end{pmatrix}=\begin{pmatrix}
0\\
0\\
0\\
0\\
\frac{1}{2}(q_0^2-q_1^2+q_2^2-q_3^2)\delta_3-(q_2q_3-q_0q_1)\delta_2\\
(q_2q_3-q_0q_1)\delta_1-(q_1q_2+q_0q_3)\delta_3\\
(q_1q_2+q_0q_3)\delta_2-\frac{1}{2}(q_0^2-q_1^2+q_2^2-q_3^2)\delta_1
\end{pmatrix}.
\]
Let us consider 
the six vector fields $X,Y_1,Y_2,[X,Y_1],[X,Y_2],[[X,Y_1],Y_1]$: we have that  the 
determinant of the 
matrix 
obtained by removing the first row from the $7\times 6$ matrix
$$(X(q,P),Y_1(q,P),Y_2(q,P),[X,Y_1](q,P),[X,Y_2](q,P),[[X,Y_1],Y_1](q,P))$$ 
is equal to $S(q,P):=S_1(q)S_2(q)S_3(q)S_4(q)S_5(P)$, where
\begin{align*}
&S_1(q):=\frac{I_2-I_3}{{32}I_2^{{3}}I_3^2}q_1, \\
&S_2(q):=(-2q_1q_2\delta_1+2q_0q_3\delta_1+q_0^2\delta_2+q_1^2\delta_2-(q_2^2+q_3^2)\delta_2), \\
&S_3(q):=(q_0(-2q_2\delta_1+2q_1\delta_2)+2q_3(q_1\delta_1+q_2\delta_2)+(q_0^2-q_1^2-q_2^2+q_3^2)\delta_3)^2, \\
&S_4(q):=(-2(q_0q_2+q_1q_3)(\delta_1^2+\delta_2^2)+((q_0^2+q_1^2-q_2^2-q_3^2)\delta_1+2(q_1q_2-q_0q_3)\delta_2)\delta_3),\\
&S_5(P):=P_1\delta_1+P_2\delta_2+P_3\delta_3=\langle P,\delta \rangle.
\end{align*}
Hence, for all $(q,P) $ such that $S(q,P)\neq0$, 
\begin{align*}
\dim\Big(\mathrm{span}&\{X(q,P),Y_1(q,P),Y_2(q,P),[X,Y_1](q,P),[X,Y_2](q,P),\\ & [[X,Y_1],Y_1](q,P) \}\Big)=6,
\end{align*}
that is, outside the set $N:=
\{(q,P)\in S^3 \times \mathbb{R}^3\mid S(q,P)=0\}$ the 
family $X,Y_1,Y_2$ 
is Lie bracket generating.

Now we are left to prove that $\mathrm{Reach}(q,P)\not\subset N$ for every $(q,P) \in N$,
and then to apply Lemma \ref{lemmino}. 
Let us start by considering the factor $S_5$ of $S$ and notice that, for any fixed $q \in S^3$, $\{S_5=0\}$ defines a surface inside $\{q\}\times \mathbb{R}^3$. Denote by $\Pi_{\mathbb{R}^3}:T(S^3\times \mathbb{R}^3)\to T\mathbb{R}^3$ the projection onto the $\mathbb{R}^3$ part of the tangent bundle.
The vector field $\Pi_{\mathbb{R}^3}X$ is tangent to $\{S_5=0\}$ when 
\[
\langle \nabla_P S_5,\Pi_{\mathbb{R}^3}X \rangle=\frac12 \langle \delta , [P , \beta P] \rangle=0, 
\]
that is, if and only if $P_3=0$ or $P_2 \delta_1-P_1 \delta_2=0$. Notice that one vector between $\Pi_{\mathbb{R}^3}Y_1,\Pi_{\mathbb{R}^3}Y_2,\Pi_{\mathbb{R}^3}Y_3$ is not tangent to $\{P_3=0\}$, otherwise
\[\mathrm{span}\{ \Pi_{\mathbb{R}^3}Y_1,\Pi_{\mathbb{R}^3}Y_2,\Pi_{\mathbb{R}^3}Y_3 \}\subset\{P_3=0 \}=
\begin{pmatrix}
0\\
0\\
1
\end{pmatrix}
^\perp.
\]
However, 
$\mathrm{span}\{ \Pi_{\mathbb{R}^3}Y_1,\Pi_{\mathbb{R}^3}Y_2,\Pi_{\mathbb{R}^3}Y_3 \}=
\delta
^\perp
,$
which would imply that $\delta$ is collinear to $(0,0,1)^T$, which is impossible since the molecule is accidentally symmetric.

Concerning the hypersurface $\{P_2 \delta_1-P_1 \delta_2=0\}$, we consider again $\Pi_{\mathbb{R}^3}X$, which is tangent to it when $\langle \nabla_P ( P_2 \delta_1-P_1 \delta_2) ,\Pi_{\mathbb{R}^3}X \rangle=0$, that is, if and only if  $P_3=0$ or $P_1 \delta_1+P_2 \delta_2=0$. We treat the second case, being  $P_3=0$ already treated. Hence, we consider the intersection
\[
\begin{cases}
P_2 \delta_1-P_1 \delta_2=0, & \\ 
P_1 \delta_1+P_2 \delta_2=0. &
\end{cases}
\]
The only solution of the system is $P_1=P_2=0$, because the molecule is accidentally symmetric.  Finally, when $P_1=P_2=0$, we consider the two-dimensional distribution $\mathrm{span}\{ \Pi_{\mathbb{R}^3}Y_1,\Pi_{\mathbb{R}^3}Y_2,\Pi_{\mathbb{R}^3}Y_3 \}$, which cannot be tangent to the $P_3$ axis. 

Summarizing, if $\delta$ is not collinear to $(0,0,1)^T$, we have
$$\mathrm{Reach}(q,P)\not\subset \{S_5=0\}, \quad  \forall (q,P) \in \{S_5=0\}.$$

To conclude, if $(q,P) \in \{S_i=0\}$, $i=1,\dots,4$, then we fix $P$ and we get two-dimensional strata $\{ q \in S^3\mid S_i(q)=0\} \subset S^3$. Now the projections of the vector fields $X,[X,Y_1],[X,Y_2],[X,Y_3]$ on the base part of the bundle span a three-dimensional vector space if $\langle P, \delta\rangle \neq 0$, as observed in \eqref{3dim}. So, by possibly steering $P$ to a point where $\langle P, \delta\rangle \neq 0$, it is possible to exit from the union of $\{S_i=0\}$. This concludes the proof of the theorem.
\end{proof}

\subsection{Reachable sets of the classical genuine symmetric-top}
Theorem \ref{genuinecla} states that each hypersurface $\{P_3=\mathrm{const}\}$ is invariant for the controlled motion. Next we prove that the restriction of system \eqref{euler1} to any such hypersurface is controllable.
\begin{theorem}\label{reachcla}
Let $I_1=I_2$ and $\delta = (0,0,\delta_3)^T$, $\delta_3\neq 0$. Then for $({g}_0,{P}_0)\in {\rm SO}(3)\times \mathbb{R}^3$, ${P}_0=(P_{01},P_{02},P_{03})$, one has 
$$\mathrm{Reach}(g_0,P_0)=\{(g,P)\in  {\rm SO}(3)\times \mathbb{R}^3\mid  P_3=P_{03}\}. $$
\end{theorem}
\begin{proof}
From Theorem \ref{genuinecla} we know that $\{P_3=\mathrm{const}\}$ is invariant. Since the drift $X$ is recurrent, it suffices to prove that system (\ref{euler1}) is Lie bracket generating on the $5$-dimensional manifold $\{P_3=\mathrm{const}\}$.

We recall from \eqref{3dim} that, if $\langle P,\delta \rangle \neq 0$, that is, if $P_3\neq 0$, we have
\begin{align*}
\dim &\Big( \mathrm{span}\{\Pi_{{\rm SO}(3)}X(g,P),\Pi_{{\rm SO}(3)}[X,Y_1](g,P),\Pi_{{\rm SO}(3)}[X,Y_2](g,P),\\ &\Pi_{{\rm SO}(3)}[X,Y_3](g,P) \} \Big)=3.
\end{align*}
Moreover, since 
$\Pi_{\mathbb{R}^3} Y_i(q,P)=(g^{-1}e_i) \times \delta $ for $i=1,2,3$,  we have that
\begin{equation}\label{localcontr}
\dim\Big( \mathrm{span}\{\Pi_{\mathbb{R}^3} Y_1,\Pi_{\mathbb{R}^3} Y_2,\Pi_{\mathbb{R}^3} Y_3 \} \Big)=2 
\end{equation}
everywhere. Thus, if $P_3\neq 0$, it follows that 
\begin{align*} 
\dim &\Big( \mathrm{span} \{X(g,P),Y_1(g,P),Y_2(g,P),Y_3(g,P),[X,Y_1](g,P),[X,Y_2](g,P),\\ &[X,Y_3](g,P) \} \Big)=5.
\end{align*}
So the system is Lie bracket generating on the manifold $\{P_3=\text{const} \neq 0\}$.

We are left to consider the case $P_3=0$. Notice that $\Pi_{\mathbb{R}^3} Y_1,\Pi_{\mathbb{R}^3} Y_2,\Pi_{\mathbb{R}^3} Y_3$ span a two-dimensional distribution for any value of $P_3$. So we consider in the quaternion parametrization the projections of $X,[X,Y_1],[X,Y_2],[[X,Y_1],X]$ on the $S^3$ part of the bundle and we obtain 
\begin{align*}
\dim& \Big( \mathrm{span} \{\Pi_{S^3}X(q,P),\Pi_{S^3}[X,Y_1](q,P),\Pi_{S^3}[X,Y_2](q,P),\\ &\Pi_{S^3}[[X,Y_1],X](q,P) \} \Big)=3, 
\end{align*}
for $P_3=0$, except when $q_3[2P_2(q_1q_2-q_0q_3)+ P_1(q_0^2+q_1^2-q_2^2-q_3^2) ]=0$.
This equation defines the union of two surfaces inside $S^3$. (Notice that we can assume $P_1 \neq 0$ and $P_2\neq 0$ because \eqref{localcontr} gives local controllability in $(P_1,P_2)$). On $\{q_3=0\}$, we have that $\Pi_{S^3}X$ is tangent if and only if $q_1P_2-q_2P_1=0$. On the curve $\gamma \subset S^3$ of equation
\[
\begin{cases}
q_3=0, & \\ 
q_1P_2-q_2P_1=0, &
\end{cases}
\]
 we can consider the two-dimensional distribution spanned by $\Pi_{S^3}[X,Y_1],\Pi_{S^3}[X,Y_2]$, $\Pi_{S^3}[X,Y_3]$, which is clearly not tangent to $\gamma$. Following Lemma \ref{lemmino}, the system is Lie bracket generating also on $\{q_3=0\}$. 
 
Analogously, on $\{2P_2(q_1q_2-q_0q_3)+ P_1(q_0^2+q_1^2-q_2^2-q_3^2)=0 \}$ we consider the vector field $\Pi_{S^3}[[[X,Y_1],X],Y_2]$ which is tangent if and only if $(q_0q_2+q_1q_3)(P_1q_0q_1+P_2q_0q_2-P_2q_1q_3+P_1q_2q_3)=0$. Again, since the distribution spanned by $\Pi_{S^3}[X,Y_1],\Pi_{S^3}[X,Y_2]$, $\Pi_{S^3}[X,Y_3]$ is two-dimensional, we can exit from the set of equations
\[
\begin{cases}
2P_2(q_1q_2-q_0q_3)+ P_1(q_0^2+q_1^2-q_2^2-q_3^2)=0, & \\ 
(q_0q_2+q_1q_3)(P_1q_0q_1+P_2q_0q_2-P_2q_1q_3+P_1q_2q_3)=0, &
\end{cases}
\]
whose strata have dimension at most one. Thus, applying again Lemma \ref{lemmino}, we can conclude that 
the restriction of
the system to the manifold $\{P_3=0\}$ is Lie bracket generating.
\end{proof}

\section{Quantum symmetric-top molecule}\label{quantum}

\subsection{Controllability of the multi-input Schr{\"o}dinger equation}\label{notations}
Let $\ell \in \mathbb{N}$ 
and $U\subset \mathbb{R}^\ell$ be a neighborhood of the origin. 
 Let $\mathcal{H}$ be an infinite-dimensional Hilbert space with scalar product $\langle \cdot, \cdot \rangle $ (linear in the first entry and conjugate linear in the second), $H, B_1,\dots,B_\ell$ be (possibly unbounded) self-adjoint operators on $\mathcal{H}$, with domains $D(H),D(B_1),\dots,D(B_\ell)$. We consider the controlled Schr{\"o}dinger equation
\begin{equation}\label{quantumcontrol}
\ii\frac{d\psi(t)}{dt}=(H+\sum_{j=1}^\ell u_j(t)B_j)\psi(t), \quad \psi(t) \in \mathcal{H}, \quad u(t) \in U.
\end{equation}
\begin{definition}
\begin{itemize}
\item We say that the operator $H$ satisfies ($\mathbb{A}1$) if it has discrete spectrum with infinitely many distinct eigenvalues (possibly degenerate).\\
Denote by $\mathcal{B}$ a Hilbert basis $(\phi_k)_{k\in \mathbb{N}}$ of $\mathcal{H}$ made of eigenvectors of $H$ associated with the family of eigenvalues $(\lambda_k)_{k\in \mathbb{N}}$ and let $\mathcal{L}$ be the set of finite linear combination of eigenstates, that is,
$$\mathcal{L}= \mathrm{span} \{\phi_k\mid k\in \mathbb{N} \}. $$
\item We say that $(H,B_1,\dots,B_\ell,\mathcal{B})$ satisfies ($\mathbb{A}2$) if $\phi_k \in D(B_j)$ for every $k \in \mathbb{N}$, $j=1,\dots,\ell$.
\item We say that $(H,B_1,\dots,B_\ell,\mathcal{B})$ satisfies ($\mathbb{A}3$) if 
$$H+\sum_{j=1}^\ell u_jB_j:\mathcal{L}\rightarrow \mathcal{H}$$
is essentially self-adjoint for every $u \in U$.
\item We say that $(H,B_1,\dots,B_\ell,\mathcal{B})$ satisfies ($\mathbb{A}$) if $H$ satisfies ($\mathbb{A}1$) and\\ $(H,B_1,\dots,B_\ell,\mathcal{B})$ satisfies  ($\mathbb{A}2$) and  ($\mathbb{A}3$).
\end{itemize}
\end{definition}

If  $(H,B_1,\dots,B_\ell,\mathcal{B})$ satisfies ($\mathbb{A}$) then, for every $(u_1,\dots,u_\ell)\in U$, $H+\sum_{j=1}^\ell u_jB_j$ generates a one-parameter group $e^{-\ii t(H+\sum_{j=1}^\ell u_jB_j)}$ inside the group of unitary operators $U(\mathcal{H})$. It is therefore possible to define the propagator $\Gamma_T^u$ at time $T$ of system (\ref{quantumcontrol}) associated with a 
piecewise constant control law $u(\cdot)=(u_1(\cdot),\dots,u_\ell(\cdot))$ by composition of flows of the type $e^{-\ii t(H+\sum_{j=1}^\ell u_jB_j)}$. 
\begin{definition}
Let  $(H,B_1,\dots,B_\ell,\mathcal{B})$ satisfy ($\mathbb{A}$).
\begin{itemize}
\item Given $\psi_0,\psi_1$ in the unit sphere $\mathcal{S}$ of $\mathcal{H}$, we say that $\psi_1$ is reachable from $\psi_0$ if there exist a time $T>0$ and a piecewise constant control law $u:[0,T]\rightarrow U$ such that $\psi_1=\Gamma_T^u(\psi_0)$. We denote by $\mathrm{Reach}(\psi_0)$ the set of reachable points from $\psi_0$.
\item We say that (\ref{quantumcontrol}) is approximately controllable if for every $\psi_0\in \mathcal{S}$ the set $\mathrm{Reach}(\psi_0)$ is dense in $\mathcal{S}$.
\end{itemize}
\end{definition}

As a byproduct of the techniques 
used to prove approximate controllability of  (\ref{quantumcontrol}) 
for our problem,
we will actually obtain a slightly stronger controllability property. For this reason, 
let us introduce the notion of module-tracker (m-tracker, for brevity) that is, a system for which any given curve can be tracked up to (relative) phases. The identification up to phases of elements of $\mathcal{H}$ in the basis $\mathcal{B}=(\phi_k)_{k\in \mathbb{N}}$ can be accomplished by the projection
$$\mathcal{M}:\psi \mapsto \sum_{k\in \mathbb{N}} |\langle \phi_k,\psi \rangle| \phi_k.$$
\begin{definition}
Let $(H,B_1,\dots,B_\ell,\mathcal{B})$ satisfy ($\mathbb{A}$). We say that system (\ref{quantumcontrol}) is an \emph{m-tracker} if, for every $r\in \mathbb{N}$, $\psi_1,\dots,\psi_r$ in $\mathcal{H}$, $\widehat{\Gamma}:[0,T]\rightarrow U(\mathcal{H})$ continuous with $\widehat{\Gamma}_0=\mathrm{Id}_{\mathcal{H}}$, and $\epsilon >0$, there exists an invertible increasing continuous function $\tau:[0,T]\rightarrow [0,T_\tau]$ and a piecewise constant control $u:[0,T_\tau]\rightarrow U$ such that
$$\|\mathcal{M}(\widehat{\Gamma}_t\psi_k)-\mathcal{M}(\Gamma_{\tau(t)}^u\psi_k)  \| <\epsilon, \qquad k=1,\dots,r,$$
for every $t\in [0,T_\tau]$.
\end{definition}

\begin{remark}
We recall that if system \eqref{quantumcontrol} is an m-tracker, then it is also approximately controllable, as noticed in \cite[Remark 2.9]{BCS}.
\end{remark}

Following \cite{BCS}
and \cite{CS},
we now introduce some objects 
that we later use
to state a 
sufficient condition for a system to be an m-tracker.
The proposed sufficient condition 
can be seen as a generalization of the main controllability result in \cite{BCS}. 
The main difference is that here, instead of testing a sequence of finite-dimensional properties on an increasing sequence of linear subspaces of $\mathcal{H}$, we test them on a sequence of overlapping finite-dimensional spaces, not necessarily ordered by inclusion. This allows the sufficient condition to be checked block-wise.

Let $\{I_j\mid j\in \mathbb{N}\}$ be a family of finite subsets of $\mathbb{N}$ such that $\cup_{j\in \mathbb{N}}I_j=\mathbb{N}$. 
Denote by $n_j$
the cardinality of $I_j$. Consider the subspaces 
\[\mathcal{M}_j:=\mathrm{span}\{\phi_n \mid n\in I_j\}\subset \mathcal{H}\]
and their associated orthogonal projections 
\[\Pi_{\mathcal{M}_j}:\mathcal{H}\ni \psi \mapsto \sum_{n\in I_j}   \langle \phi_n,\psi \rangle \phi_n \in \mathcal{H}.
\]

 Given a linear operator $Q$ on $\mathcal{H}$ we identify the linear operator $\Pi_{\mathcal{M}_j} Q \Pi_{\mathcal{M}_j}$ 
preserving 
$\mathcal{M}_j$ 
with its 
complex matrix representation with respect to the basis 
$(\phi_n)_{n\in I_j}$.
The  set $\Sigma_j=\{|\lambda_l-\lambda_{l'}|\mid l,l'\in I_j \}$ is then the collection of the spectral gaps of 
$\Pi_{\mathcal{M}_j} H \Pi_{\mathcal{M}_j}$.
We define $B_i^{(j)}:= \Pi_{\mathcal{M}_j} B_i \Pi_{\mathcal{M}_j}$ 
for every $i=1,\dots,\ell$.

If the element $(B_i)_{l,k}$ is different from zero, then a control $u_i$ oscillating at frequency $|\lambda_l-\lambda_k|$ induces a population transfer between the states $\phi_l$ and $\phi_k$ (\cite{chambrion}). The dynamics of such a population transfer depend on the other pairs of states  $\phi_{l'}$, $\phi_{k'}$ having the same spectral gap and whose corresponding element $(B_i)_{l',k'}$ is different from zero. We are interested in controlling the induced 
population dynamics within a space ${\cal M}_j$. 
This motivates the definition of the sets 
\begin{align*}
\Xi_j^0=\{(\sigma,i)\in \Sigma_j \times \{1,\dots,\ell\} \mid  \mbox{}&
(B_i)_{l,k}=0\ \mbox{for every }l\in \mathbb{N},\ k\in \mathbb{N}\setminus I_j\\
&\mbox{ such that }|\lambda_l-\lambda_k|=\sigma 
 \},
\end{align*}
and
\begin{align*} 
\Xi_j^1=\{(\sigma,i)\in \Sigma_j \times \{1,\dots,\ell\} \mid \mbox{}& (B_i)_{l,k}=0\ \mbox{for every }l\in I_j,\ k\in \mathbb{N}\setminus I_j\\
&\mbox{ such that }|\lambda_l-\lambda_k|=\sigma  \}.
\end{align*}
While the set $\Xi_j^1$ compares only with pairs of states $\phi_{l},\phi_k$ with $\phi_l$ in ${\cal M}_j$, such a requirement is not present in the definition if $\Xi_j^0$. This means that for $(\sigma,i)\in \Xi_j^0$ the induced population dynamics obtained by a control $u_i$ oscillating at frequency $\sigma$ 
not only does not produce population transfer out of ${\cal M}_j$, but also
is trivial within the orthogonal complement to ${\cal M}_j$.

For every $\sigma \geq 0$, and every square matrix $M$ of dimension $m$, let
$$\mathcal{E}_\sigma(M)=(M_{l,k}\delta_{\sigma,|\lambda_l-\lambda_k| })_{l,k=1,\dots,m}, $$
where $\delta_{l,k}$ is the Kronecker delta. The $n_j\times n_j$ matrix $\mathcal{E}_\sigma(B_i^{(j)})$ corresponds to the activation in $B_i^{(j)}$ of the spectral gap $\sigma\in \Sigma_j$: every element is $0$ except for the $(l,k)$-elements such that $|\lambda_l-\lambda_k|=\sigma$.
A control $u_i$ oscillating  at frequency $\sigma$ can induce the dynamics in ${\cal M}_j$ described by the matrix $\mathcal{E}_\sigma(B_i^{(j)})$, and also, by phase modulation,
those described
by the matrix 
$W_\xi$, $\xi \in S^1\subset \mathbb{C}$, 
defined by 
\begin{equation}
(W_\xi(M))_{l,k}=\begin{cases}
\xi M_{l,k}, & \lambda_l< \lambda_k, \\ 
0, & \lambda_l = \lambda_k, \\ 
\bar{\xi} M_{l,k}, & \lambda_l > \lambda_k.  
\end{cases} 
\end{equation}
Let us 
consider 
the sets of excited modes 
\begin{equation}\label{eq:modes}
\nu_{j}^s:= \{W_\xi(\mathcal{E}_\sigma(\ii B_i^{(j)})) \mid (\sigma,i)\in \Xi_{j}^s, \sigma \neq 0, \xi \in S^1 \}, \quad s=0,1.
\end{equation}
Notice that $\nu_{j}^0 \subset \nu_{j}^1 \subset \mathfrak{su}(n_j)$. Indeed, we have the following picture: 
$$\mathcal{E}_\sigma(\ii B_i^{(j)})\in\nu_j^{0} \Rightarrow \mathcal{E}_\sigma(\Pi_{j-1,j,j+1}\ii B_i\Pi_{j-1,j,j+1})=\left[
\begin{array}{c|c|c}
0 & 0 &0\\
\hline
0 & \mathcal{E}_\sigma(\ii B_i^{(j)}) &0\\
\hline
0 & 0 &0
\end{array}
\right] $$
$$\mathcal{E}_\sigma(\ii B_i^{(j)})\in\nu_j^{1} \Rightarrow \mathcal{E}_\sigma(\Pi_{j-1,j,j+1}\ii B_i\Pi_{j-1,j,j+1})=\left[
\begin{array}{c|c|c}
* & 0 &*\\
\hline
0 & \mathcal{E}_\sigma(\ii B_i^{(j)}) &0\\
\hline
* & 0 &*
\end{array}
\right] $$
where $\Pi_{j-1,j,j+1}$ denotes the projection onto $\mathcal{M}_{j-1}\oplus\mathcal{M}_{j}\oplus\mathcal{M}_{j+1}$.

 We denote by $\mathrm{Lie}(\nu_{j}^s)$ the Lie subalgebra  of $\mathfrak{su}(n_j)$ generated by the matrices in $\nu_{j}^s$, $s=0,1$,  and define 
$\mathcal{T}_j$ as the minimal ideal of $ \mathrm{Lie}(\nu_j^1)$ containing $\nu_j^0$.

Finally, we introduce the graph $\mathcal{G}$ with vertices $\mathcal{V}=\{I_j\mid j\in \mathbb{N}\}$ and edges $\mathcal{E}=\{(I_j,I_k)\mid j,k\in\mathbb{N},\;I_j\cap I_k\neq \emptyset\}$. We are now in a position to state a new sufficient condition for a system to be an m-tracker, and thus, approximately controllable.
\begin{theorem}\label{LGTC}
Assume that  ($\mathbb{A}$) holds true. If the graph $\mathcal{G}$ is connected and $\mathcal{T}_j=\mathfrak{su}(n_j)$ for every $j\in \mathbb{N}$, then \eqref{quantumcontrol} is an m-tracker.
\end{theorem}
\begin{proof}
The proof works by applying Theorem~2.8 in \cite{BCS}, which 
guarantees that \eqref{quantumcontrol} is an m-tracker if a suitable condition, called Lie--Galerkin tracking condition (\cite[Definition 2.7]{BCS}), holds true. In terms of the notation introduced here, the Lie--Galerkin tracking condition is true if there exists a sequence $\{\widetilde{I}_j\mid j\in \mathbb{N}\}$ of finite subsets of $\mathbb{N}$, strictly increasing with respect to the inclusion, such that $\cup_{j\in \mathbb{N}}\widetilde{I}_j=\mathbb{N}$ and $\mathcal{T}_j=\mathfrak{su}(n_j)$ for every $j\in \mathbb{N}$.

Up to reordering the sets $I_j$, we can assume that 
\begin{equation}\label{eq:ordering}
I_{j+1}\cap (\cup_{k=1}^j I_k)\not= \emptyset,\qquad \forall j\in\mathbb{N}.
\end{equation} 
For $j\in \mathbb{N}$, let $\widetilde{I}_j=\cup_{i=1}^jI_i$ and $  {\cal Z}_j=
\sum_{k=1}^j \mathcal{M}_{k}.$

The Lie--Galerkin tracking condition holds true if
\begin{equation}\label{induction}
\mathrm{Lie}(\cup_{j=1}^m\widetilde{\mathcal{T}}_j)=\mathfrak{su}(\dim({\cal Z}_m)),\qquad m\in\mathbb{N},
\end{equation}
where the set of operators $\widetilde{\mathcal{T}}_j$ is obtained similarly to $\mathcal{T}_j$, 
replacing 
$\nu_j^s$, $s=0,1$, by 
\begin{equation*}
\{W_\xi(\mathcal{E}_\sigma(\ii \Pi_{{\cal Z}_m
}B_i\Pi_{{\cal Z}_m
})) \mid (\sigma,i)\in \Xi_{j}^s, \sigma \neq 0, \xi \in S^1 \}, \quad s=0,1.
\end{equation*}

We proceed by induction on $m$. For $m=1$, \eqref{induction} is true, since we have that $\mathrm{Lie}(\mathcal{T}_1)=\mathcal{T}_1=\mathfrak{su}(n_1)=\mathfrak{su}( \dim(\mathcal{Z
}_1))$.
Assume now that \eqref{induction} is true for $m$, and consider the vertex $I_{m+1}\in {\cal V}$. 
Consider $t,p\in \cup_{j=1}^{m+1}I_{j}$ and let us prove that $G_{t,p}:= e_{t,p}- e_{p,t}$ is in 
$\mathrm{Lie}(\cup_{j=1}^{m+1}\widetilde{\mathcal{T}}_j)$, 
where $e_{a,b}$ is the matrix with all entries equal to 0 except for the one in row $a$ and column $b$, which is equal to 1 (and the indices in $\cup_{j=1}^{m+1}I_{j}$ are identified with the elements of $\{1,\dots,\dim({\cal Z}_{m+1}
)\}$).
Decomposing ${\cal Z}_{m+1}$ as a direct orthogonal sum $V_1\oplus ({\cal Z}_{m}\cap {\cal M}_{m+1}) \oplus V_2$ with $V_1\subset  {\cal Z}_{m}$ and $V_2\subset {\cal M}_{m+1}$, 
a matrix in  $\widetilde{\mathcal{T}}_{m+1}$ has the 
form 
\[\left[
\begin{array}{c|c|c}
0 & 0 &0\\
\hline
0 & Q_{11} &Q_{12}\\
\hline
0 & Q_{21} &Q_{22}
\end{array}
\right],\qquad \left[
\begin{array}{c|c}
Q_{11} &Q_{12}\\
\hline
 Q_{21} &Q_{22}
\end{array}
\right]\in \mathfrak{su}(n_{m+1}),\]
as it follows from the definition of $\Xi_j^0$ and $\Xi_j^1$ and the fact that $\widetilde{\mathcal{T}}_{m+1}$ is the ideal generated by $\nu_j^0$ inside ${\rm Lie}(\nu_j^1)$.
Similarly, a matrix in $\cup_{j=1}^{m}\widetilde{\mathcal{T}}_j$ has the form 
\[\left[
\begin{array}{c|c|c}
Q_{11}&Q_{12} & 0 \\
\hline
Q_{21}&Q_{22} & 0 \\
\hline
0 & 0 &0
\end{array}
\right],\qquad \left[
\begin{array}{c|c}
Q_{11} &Q_{12}\\
\hline
 Q_{21} &Q_{22}
\end{array}
\right]\in \mathfrak{su}(\dim({\cal Z}_{m}
)).\]

If $t,p\in \cup_{j=1}^mI_{j}$ or $t,p\in I_{m+1}$ the conclusion follows from the induction hypothesis and the identity $\mathcal{T}_{m+1}=\mathfrak{su}(n_{m+1})$. 
Let then $t\in I_{m+1}\setminus (\cup_{j=1}^mI_{j})$ and $p\in \cup_{j=1}^mI_{j}$. 
Fix, moreover, $r\in I_{m+1}\cap (\cup_{j=1}^mI_{j})$, whose existence is guaranteed by \eqref{eq:ordering}. 
Again by the induction hypothesis and the identity $\mathcal{T}_{m+1}=\mathfrak{su}(n_{m+1})$, 
we have that 
$G_{p,r}
$ and $G_{r,t}
$ are in $\mathrm{Lie}(\cup_{j=1}^{m+1}\widetilde{\mathcal{T}}_j)$.
The bracket $[G_{p,r},G_{r,t}]=G_{p,t}$ is therefore also in $\mathrm{Lie}(\cup_{j=1}^{m+1}\widetilde{\mathcal{T}}_j)$. By similar arguments, we deduce that  every  element of a basis of $\mathfrak{su}(\dim({\cal Z}_{m+1}
))$ is in $\mathrm{Lie}(\cup_{j=1}^{m+1}\widetilde{\mathcal{T}}_j)$.
 \end{proof}

\subsection{The Schr{\"o}dinger equation of a symmetric-top subject to electric fields}
We recall in this section some general facts about Wigner functions and the theory of angular momentum in quantum mechanics (see, for instance, \cite{Varshalovich,gordy}).

We use Euler's angles $(\alpha,\beta,\gamma)\in[0,2\pi)\times[0,\pi]\times [0,2\pi)$ to describe the configuration space ${\rm SO}(3)$ of the molecule. More precisely, the coordinates of a vector change from the body fixed frame $a_1,a_2,a_3$ to the space fixed frame $e_1,e_2,e_3$ via three rotations
\begin{equation}\label{3rot}
\begin{pmatrix}
X\\
Y\\
Z
\end{pmatrix}=R_{e_3}(\alpha)R_{e_2}(\beta)R_{e_3}(\gamma)\begin{pmatrix}
x\\
y\\
z
\end{pmatrix}=:R(\alpha,\beta,\gamma)\begin{pmatrix}
x\\
y\\
z
\end{pmatrix}
\end{equation}
where $(x,y,z)^T$ are the coordinates of the vector in the body fixed frame, $(X,Y,Z)^T$ are the coordinates of the vector in the space fixed frame and $R_{e_i}(\theta)\in {\rm SO}(3)$ is the rotation of angle $\theta$ around the axis $e_i$. The explicit expression of the matrix $R(\alpha,\beta,\gamma)\in {\rm SO}(3)$ is 
\begin{equation}\label{rotation}
R=\begin{pmatrix}
\cos\alpha \cos\beta \cos\gamma-\sin\alpha \sin\gamma & -\cos\alpha \cos\beta \sin\gamma-\sin\alpha \cos\gamma & \cos\alpha \sin\beta\\
\sin\alpha \cos\beta \cos\gamma+\cos\alpha \sin\gamma &-\sin\alpha \cos\beta \sin\gamma+\cos\alpha \cos\gamma & \sin\alpha \sin\beta\\
-\sin\beta \cos\gamma& \sin\beta \sin\gamma& \cos\beta
\end{pmatrix}.
\end{equation}

In Euler coordinates, the angular momentum operators are given by
 \begin{equation}\label{wigner}
\begin{cases}
\begin{aligned}
J_1&=\ii\cos\alpha \cot\beta \dfrac{\partial}{\partial \alpha}+\ii\sin\alpha \dfrac{\partial}{\partial \beta} -\ii\dfrac{\cos\alpha}{\sin \beta} \dfrac{\partial}{\partial \gamma}   ,  \\ 
J_2&=\ii\sin\alpha \cot\beta \dfrac{\partial}{\partial \alpha}-\ii\cos\alpha \dfrac{\partial}{\partial \beta} -\ii\dfrac{\sin\alpha}{\sin \beta} \dfrac{\partial}{\partial \gamma}  ,  \\
J_3&=-\ii\dfrac{\partial}{\partial \alpha}.
\end{aligned}
\end{cases}
\end{equation}
These are linear operators acting on the Hilbert space $L^2({\rm SO}(3))$, self-adjoint with respect to the Haar measure $\frac{1}{8}d\alpha d\gamma \sin\beta d\beta$. Using (\ref{wigner}), the self-adjoint operator $P_3:=-\ii\frac{\partial}{\partial \gamma}$ can be written as $P_3=\sin\beta \cos\alpha J_1+\sin\beta \sin\alpha J_2+\cos\beta J_3$, that is,
\[
P_3=\sum_{i=1}^3R_{i3}(\alpha,\beta,\gamma)J_i,
\]
where $R=(R_{ij})_{i,j=1}^3$ is given in (\ref{rotation}).

 In the same way we define $P_1=\sum_{i=1}^3R_{i1}(\alpha,\beta,\gamma)J_i,P_2=\sum_{i=1}^3R_{i2}(\alpha,\beta,\gamma)J_i$. The operators $J_i$ and $P_i$, $i=1,2,3$, are the angular momentum operators expressed in the fixed and in the body frame, respectively. Finally, we consider the square norm operator $J^2:=J_1^2+J_2^2+J_3^2=P_1^2+P_2^2+P_3^2$. Now, $J^2,J_3,P_3$ can be considered as the three commuting observables needed to describe the quantum motion of a molecule. 
Indeed, $[J^2,J_3]=[J^2,P_3]=[J_3,P_3]=0$, and hence there exists an orthonormal Hilbert basis of $L^2({\rm SO}(3))$ which diagonalizes simultaneously $J^2,J_3$ and $P_3$. In terms of Euler coordinates, this basis is made by the so-called Wigner functions 
\begin{equation}\label{explicit}
 D_{k,m}^j(\alpha,\beta,\gamma):=e^{\ii(m\alpha+k\gamma)}d_{k,m}^j(\beta), \qquad j\in \mathbb{N},\quad k,m=-j,\dots,j,
 \end{equation}
 where the function $d^j_{k,m}$ solves 
 a suitable Legendre differential equation, obtained by separation of variables (see, e.g., \cite[Section 2.5]{gordy} for the separation of variables ansatz and \cite[Chapter 4]{Varshalovich} for a detailed description of the properties of these functions).

Summarizing, the family of Wigner functions $\{ D_{k,m}^j\mid j\in \mathbb{N},k,m=-j,\dots,j \}$ forms an orthonormal Hilbert basis for $L^2({\rm SO}(3))$. 
Moreover,
\[
J^2 D_{k,m}^j=j(j+1) D_{k,m}^j, \quad J_3 D_{k,m}^j=mD_{k,m}^j, \quad P_3D_{k,m}^j=kD_{k,m}^j.
\]

Thus, $m$ and $k$ are the quantum numbers which correspond to the projections of the angular momentum on the third axis of, respectively, the fixed and the body frame.

The rotational Hamiltonian of a molecule is $H=\frac{1}{2}\Big(\frac{P_1^2}{I_1}+\frac{P_2^2}{I_2}+\frac{P_3^2}{I_3}\Big)$, which is seen here as a self-adjoint operator acting on the Hilbert space $L^2({\rm SO}(3))$. 
From now on, we impose the symmetry relation 
$I_1=I_2$, which implies that $H=\frac{J^2}{2I_2}+\Big(\frac{1}{2I_3}-\frac{1}{2I_2}\Big)P_3^2$. Thus, 
\begin{equation}\label{spectrum}
HD_{k,m}^j=\Big(\dfrac{j(j+1)}{2I_2}+\Big(\dfrac{1}{2I_3}-\dfrac{1}{2I_2}\Big)k^2\Big)D_{k,m}^j=:E_k^jD_{k,m}^j.
\end{equation}
Hence, the Wigner functions are the eigenfunctions of $H$.
Since the eigenvalues 
of $H$ do not depend on $m$, the energy level $E_k^j$ is $(2j+1)$-degenerate with respect to $m$.
This property is common to every molecule in nature: the spectrum $\sigma(H)$ does not depend on $m$, just like in classical mechanics kinetic energy does not depend on the direction of the angular momentum. 
Moreover, when $k \neq 0$ the energy level $E_k^j$ is also $2$-degenerate with respect to $k$. This extra degeneracy is actually a characterizing property of 
symmetric molecules. 
Breaking this $k$-symmetry will be one important feature of our controllability analysis. 

The interaction Hamiltonian between the dipole $\delta$ inside the molecule and the external electric field in the direction $e_i$, $i=1,2,3$, is given by the Stark effect (\cite[Chapter 10]{gordy})
\[
B_i(\alpha,\beta,\gamma)=-\langle R(\alpha,\beta,\gamma) \delta, e_i\rangle,
\]
seen as a multiplicative self-adjoint operator acting on $L^2({\rm SO}(3))$. Thus, the rotational Schr{\"o}dinger equation for a symmetric-top molecule subject to three orthogonal 
electric fields reads
\begin{equation}\label{schro}
\ii\dfrac{\partial}{\partial t} \psi(\alpha,\beta,\gamma;t)= H\psi(\alpha,\beta,\gamma;t)+\sum_{l=1}^3u_l(t)B_l(\alpha,\beta,\gamma)\psi(\alpha,\beta,\gamma;t), 
\end{equation}
with $\psi(t) \in L^2({\rm SO}(3))$ and $u(t) \in U$, for some neighborhood $U$ of $0$ in $\mathbb{R}^3$.

\subsection{Non-controllability of the quantum genuine symmetric-top}

We recall that the genuine symmetric-top molecule is a symmetric rigid body with electric dipole $\delta$ along the symmetry axis:
$\delta=(0,0,\delta_3)^T$ in the principal axis frame on the body. We then introduce the subspaces $S_k:=\overline{\mathrm{span}}\{D_{k,m}^j \mid j\in \mathbb{N}, m=-j,\dots,j\}$, where $\overline{\mathrm{span}}$ denotes the closure of the linear hull in $L^2({\rm SO}(3))$.
\begin{theorem}\label{genuine}
The quantum number $k$ is invariant in the controlled motion of the genuine symmetric-top molecule. That is, if $I_1=I_2$ and $\delta=(0,0,\delta_3)^T$, the subspaces $S_k$ are invariant for any propagator of the Schr{\"o}dinger equation (\ref{schro}).
\end{theorem}
\begin{proof}
We have to show that $H$ and $B_1$, $B_2$, $B_3$, 
do not couple different levels of $k$, that is,
\begin{equation}\label{alternativa}
\begin{cases}
\langle D_{k,m}^j,\ii H D_{k',m'}^{j'}\rangle_{L^2({\rm SO}(3))}=0, \quad k \neq k',    \\
\langle D_{k,m}^j, \ii B_l D_{k',m'}^{j'}\rangle_{L^2({\rm SO}(3))}=0, \quad k \neq k',\; l=1,2,3.
\end{cases}
\end{equation}
The first equation of (\ref{alternativa}) is obvious since the orthonormal basis $\{D_{k,m}^j\}$ diagonalizes $H$.
Under the genuine symmetric-top assumption, the second equation of (\ref{alternativa}) is also true:  for $l=1$ and $k\neq k'$ we compute
\begin{align*}
\langle D_{k,m}^j,&\ii B_1  D_{k',m'}^{j'}\rangle_{L^2({\rm SO}(3))}\\  ={}& -\int_0^{2\pi}d\alpha\int_0^{2\pi}d\gamma\int_0^{\pi}d\beta \sin(\beta) D_{k,m}^j(\alpha,\beta,\gamma)\ii B_1(\alpha,\beta,\gamma)\overline{D_{k',m'}^{j'}}(\alpha,\beta,\gamma)  \\
={}&\ii \delta_3\left(\int_0^{2\pi}d\gamma e^{\ii k\gamma}e^{-\ii k'\gamma}\right)\left(\int_0^{2\pi}d\alpha\cos(\alpha) e^{\ii m\alpha}e^{-\ii m'\alpha}\right)
\\ &\left(\int_0^{\pi}d\beta \sin^2(\beta)d_{k,m}^j(\beta)\overline{d_{k',m'}^{j'}}(\beta)\right) =0,
\end{align*}
using the orthogonality of the functions $e^{ik\gamma}$ and the explicit form (\ref{rotation}) of the matrix $R$, which yields 
\[B_1(\alpha,\beta,\gamma)=-\langle \begin{pmatrix}
0\\
0\\
\delta_3
\end{pmatrix},R^{-1}(\alpha,\beta,\gamma)\begin{pmatrix}
1\\
0\\
0
\end{pmatrix}\rangle=-\delta_3 \sin \beta \cos\alpha.\]
 The computations for $l=2,3$ are analogous, since the multiplicative potentials $B_l$ do not depend on $\gamma$.
\end{proof}

\begin{remark}
Equation (\ref{alternativa}) also shows that, for a genuine symmetric-top, the third component of the angular momentum $P_3$ commutes with $H$ and $B_l$, $l=1,2,3$, hence
 \[
 \Big[P_3,H+\sum_{l=1}^3u_lB_l\Big]=0, \quad \forall u\in U.       
 \]
 Thus, 
 $\langle \psi(t),P_3\psi(t)\rangle$ is a conserved quantity, where $\psi$ is the solution of $(\ref{schro})$.
\end{remark}

\subsection{Controllability of the quantum accidentally symmetric-top}
So far we have studied the dynamics of a symmetric-top molecule with electric dipole moment along its symmetry axis and  we have proven that its dynamics are trapped in 
the eigenspaces of $P_3$. 

Nevertheless, for applications to molecules charged in the laboratory, or to particular molecules present in nature such as $D_2S_2$ (Figure \ref{d2s2}) or $H_2S_2$, it is interesting to consider also the case in which the dipole 
 is not along 
the symmetry axis: this case is called the \emph{accidentally symmetric molecule}.

Under a non-resonance condition, we are going to prove that, if the dipole moment is not orthogonal to the symmetry axis of the molecule, the rotational dynamics of an accidentally symmetric-top are approximately controllable. To prove this statement, we are going to apply 
Theorem \ref{LGTC} 
to \eqref{schro}.

\begin{theorem}\label{rare}
Assume that $I_1=I_2$ and $\frac{I_2}{I_3}\notin \mathbb{Q}$. If $\delta = (\delta_1,\delta_2,\delta_3)^T$ is such that $\delta \neq (0,0,\delta_3)^T$ and $\delta \neq (\delta_1,\delta_2,0)^T$, then system 
(\ref{schro}) is an m-tracker, and in particular approximately controllable.
\end{theorem}
\begin{proof}
First of all, one can check, for example in \cite[Table 2.1]{gordy}, that the pairings induced by the interaction Hamiltonians satisfy
\begin{equation}\label{rules}
\langle D_{k,m}^j , \ii B_l D_{k',m'}^{j'} \rangle=0,
\end{equation}
when $|j'-j|>1$, or $|k'-k|>1$ or $|m'-m|>1$, for every $l=1,2,3$. Equation (\ref{rules}) is the general form of the so-called selection rules.

We then define 
for every $j \in \mathbb{N}$ the set $I_j:=\{\rho(l,k,m)\mid l=j,j+1, \; k,m=-l,\dots,l\}\subset \mathbb{N}$, where $\rho: \{(l,k,m) \mid l\in \mathbb{N},k,m=-l,\dots,l\}\rightarrow \mathbb{N}$ is the lexicographic ordering.
The graph $\mathcal{G}$ whose vertices are the sets $I_j$ and whose edges are 
$\{(I_j,I_{j'})\mid I_j\cap I_{j'} \neq \emptyset\}=\{ (I_j,I_{j+1})\mid j\in\mathbb{N}\}$ is 
linear. 
In order to apply Theorem~\ref{LGTC} we 
shall 
consider the projection of \eqref{schro} onto each space $\mathcal{M}_j:=\mathcal{H}_j \oplus \mathcal{H}_{j+1}$, where $\mathcal{H}_l:= \mathrm{span}\{D_{k,m}^l \mid k,m=-l,\dots,l\}$. The dimension of $\mathcal{M}_j$ is $(2j+1)^2+(2(j+1)+1)^2$, and we identify $\mathfrak{su}(\mathcal{M}_j)$ with $\mathfrak{su}((2j+1)^2+(2(j+1)+1)^2)$. 

According to \eqref{rules}, the three types of spectral gaps in $\mathcal{M}_j$, $j\in \mathbb{N}$, which we should consider are
\begin{equation}\label{usami}
\lambda_k^j:= |E_{k+1}^{j+1}-E_k^j| =\Big| \frac{j+1}{I_2}+\Big(\frac{1}{2I_3}-\frac{1}{2I_2}\Big)(2k+1)\Big|, \quad k=-j,\dots,j, 
\end{equation}
 corresponding to pairings for which both $j$ and $k$ move (see Figure \ref{lambda}), 
 \begin{equation}\label{usami1}
 \eta_k:= |E_{k+1}^j-E_k^j| =\Big| \Big(\dfrac{1}{2I_3}-\dfrac{1}{2I_2}\Big)(2k+1) \Big|, \quad k=-j,\dots,j,  
 \end{equation}
and 
\begin{equation}\label{usami2}
\sigma^j:= |E_k^{j+1}-E_k^j| = \dfrac{j+1}{I_2}, \quad k=-j,\dots,j,
\end{equation}
for which, respectively,  only $k$ or $j$ moves (see, 
Figures \ref{transitionsfigure}\subref{eta} and \ref{transitionsfigure}\subref{sigma}).
\begin{figure}[ht!]\begin{center}
\includegraphics[width=0.5\linewidth, draft = false]{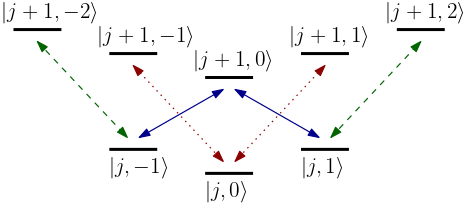}
\caption{Graph of the transitions associated with the frequency $\lambda_k^j$ between eigenstates $| j,k\rangle=| j,k,m\rangle:=D_{k,m}^j$ ($m$ fixed). Same-shaped arrows correspond to equal spectral gaps.} \label{lambda}
\end{center}\end{figure}
\begin{figure}[ht!]
\subfigure[]{
\includegraphics[width=0.47\linewidth, draft = false]{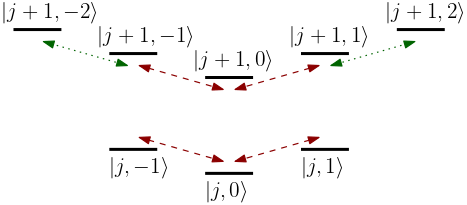} \label{eta} }
\subfigure[]{
\includegraphics[width=0.47\linewidth, draft = false]{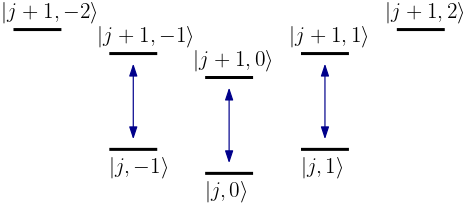} \label{sigma} }
\caption{Transitions between states: 
\subref{eta} at frequency $\eta_k$; \subref{sigma} at frequency $\sigma^j$. Same-shaped arrows correspond to equal spectral gaps.} \label{transitionsfigure}\end{figure}

We now classify the spectral gaps 
in terms of the sets $\Xi_j^0$ and $\Xi_j^1$ introduced in Section~\ref{notations}.
\begin{lemma}
\label{ext}
Let $I_2/I_3\notin \mathbb{Q}$. Then $(\lambda_k^j,l), (\sigma^j,l)\in \Xi_j^0$, and $(\eta_k,l) \in \Xi_j^1$, for all $k=-j,\dots,j$, $l=1,2,3$.
\end{lemma}
\begin{proof}
Because of the selection rules \eqref{rules}, we only need to check if there are common spectral gaps in the spaces $\mathcal{M}_j$ and $\mathcal{M}_{j'}$ for $j'\ne j$.

We start by 
proving that $(\lambda_k^j,l), (\sigma^j,l)\in \Xi_j^0$ by showing that 
a spectral gap of the type $\lambda_k^j$ (respectively, $\sigma^j$)
is different from any  spectral gap of the type $\lambda_{k'}^{j'}$, $\sigma^{j'}$, or 
$\eta_{k'}$ unless $\lambda_k^j=\lambda_{k'}^{j'}$ and $(k,j)=(k',j')$ (respectively, $\sigma^j=\sigma^{j'}$ and $j=j'$).

Using the explicit structure of the spectrum (\ref{spectrum}), any spectral gap of the type
$\lambda_{k'}^{j'}$, $\sigma^{j'}$, or 
$\eta_{k'}$ can be written as 
\[ \Big|\frac{q_1}{I_2}+q_2\Big(\dfrac{1}{I_3}-\dfrac{1}{I_2}\Big)\Big|,\qquad q_1,q_2\in \mathbb{Q}.\]
 Since, moreover,
  $\frac {1}{I_2}$ and $\left(\frac1{I_3}-\frac1{I_2}\right)$ are $\mathbb{Q}$-linearly independent,
  one easily deduces that, indeed,  $(\lambda_k^j,l), (\sigma^j,l)\in \Xi_j^0$.

Notice that the gaps of the type 
$\eta_k$ correspond to internal pairings  in the spaces ${\cal H}_j$. Henceforth, 
in order to prove that 
$(\eta_k,l) \in \Xi_j^1$ it is enough to check that $\eta_k$ is different from any gap of the type 
 $\lambda_{k'}^j,\sigma^j$. This fact has already been noticed 
 in the proof of the first part of the statement. The proof of the lemma is then concluded. 
 \end{proof}

Next, we introduce the family of 
excited modes 
associated with the spectral gap $\lambda_k^j$, that is,
\[
\mathcal{F}_j:=\{\mathcal{E}_{\lambda_k^j}(\ii B_l), W_\ii(\mathcal{E}_{\lambda_k^j}(\ii B_l)) \mid  l=1,2,3, \; k=-j,\dots,j \}, 
\]
where the operators $\mathcal{E}_\mu$ and  $W_\xi$ 
are defined in 
Section~\ref{notations}, and where, with a slight abuse of notation, we write $B_l$ instead of $\Pi_{\mathcal{M}_j}B_l\Pi_{\mathcal{M}_j}$.
Notice that $\mathcal{F}_j\subset \nu_j^0$ as it follows from Lemma~\ref{ext}, where $\nu_j^0$
is defined as in \eqref{eq:modes}.

In order to write down the matrices in $\mathcal{F}_j$, we need to study the resonances 
between the spectral gaps inside $\mathcal{M}_j$. 
We claim that there are no 
internal resonances except those due to 
the degeneracy
$E_k^j=E_{-k}^j$. 
Indeed, we 
already noticed in Lemma~\ref{ext}
that a spectral gap of the type $\lambda_k^j$ 
is different from any  spectral gap of the type $\lambda_{k'}^{j'}$, $\sigma^{j'}$, or 
$\eta_{k'}$ unless $\lambda_k^j=\lambda_{k'}^{j'}$ and $(k,j)=(k',j')$.
We collect in the lemma below also the similar observations that 
$\sigma^j$ is different from any  spectral gap of the type $\lambda_{k'}^{j'}$, $\sigma^{j'}$, or 
$\eta_{k'}$ unless
$\sigma^j=\sigma^{j'}$ and $j=j'$, and that $\eta_{k}\ne \eta_{k'}$ if $k\ne k'$. 

\begin{lemma}
\label{important}
Let $I_2/I_3\notin \mathbb{Q}$.  Then
\begin{enumerate}
\item\textbf{$\lambda_k^j$-resonances:} the equation 
\begin{equation*}
 |E_{k+1}^{j+1}-E_k^j|=|E_{s+h}^{j''}-E_{s}^{j'}|,\ j\le j'\leq j''\le j+1,\  -j'\le s\le j',\ h\in \{-1,0,1\},
 \end{equation*}

implies that $j'=j$, $j''=j+1$, $s=\pm k$, $s+h=\pm(k+1)$;
\item\textbf{$\eta_k$-resonances:} the equation 
\begin{equation*}
 |E_{k+1}^j-E_k^j|=|E_{s+h}^{j''}-E_{s}^{j'}|,\ j\le j'\leq j''\le j+1,\  -j'\le s\le j',\ h\in \{-1,0,1\}, \end{equation*}
implies that $j'=j''=j$ or $j'=j''=j+1$ and $s=\pm k$, $s+h=\pm(k+1)$;
\item\textbf{$\sigma^j$-resonances:} the equation 
\begin{equation*}
 |E_k^{j+1}-E_k^j|=|E_{s+h}^{j''}-E_{s}^{j'}|,\ j\le j'\leq j''\le j+1,\  -j'\le s\le j',\ h\in \{-1,0,1\}, \end{equation*}
implies that $j'=j$, $j''=j+1$, $h=0$, $s=\pm k$.
\end{enumerate}
\end{lemma}

Denote by $\mathrm{L}_j:=\mathrm{Lie}(\mathcal{F}_j)$ the Lie algebra generated by the matrices in $\mathcal{F}_j$. Let us introduce the generalized Pauli matrices 
$$G_{j,k}=e_{j,k}-e_{k,j}, \quad F_{j,k}=\ii e_{j,k}+\ii e_{k,j},\quad D_{j,k}=\ii e_{j,j}-\ii e_{k,k},$$
where $e_{j,k}$ denotes the $(2j+1)^2+(2(j+1)+1)^2$-square matrix whose entries are all zero, except the one at row $j$ and column $k$, which is equal to $1$. Consider again the lexicographic ordering $\rho:\{(l,k,m) \mid l=j,j+1,\;k,m=-l,\dots,l\}\rightarrow \mathbb{N}$. By a slight abuse of notation, also set $e_{(l,k,m),(l',k',m')}=e_{\rho(l,k,m),\rho(l',k',m')}$. The analogous identification can be used to define $G_{(l,k,m),(l',k',m')}, F_{(l,k,m),(l',k',m')}, D_{(l,k,m),(l',k',m')}$.
The next proposition tells us how the elements in $\mathrm{L}_j$ look like. For a
 proof, see Appendix~\ref{appendixA}.

\begin{proposition}
\label{propA}
Let $m=-j,\dots,j$ and $k=-j,\dots,j$ with $k\neq 0$.
Then the matrices $X_{(j,k,m),(j+1,k+1,m)}-X_{(j,-k,m),(j+1,-k-1,m)}$ and $X_{(j,k,m),(j+1,k+1,m\pm1)}-X_{(j,-k,m),(j+1,-k-1,m\pm1)}$ are in $\mathrm{L}_j$, where $X\in \{G,F\}$.
\end{proposition}

To break the degeneracy between $k$ and $-k$ which appears in the 
matrices that we found in 
Proposition~\ref{propA}, and obtain all the elementary matrices that one needs to generate $\mathfrak{su}(\mathcal{M}_j)$, we need to exploit the other two types of spectral gaps that we 
have introduced in \eqref{usami1} and \eqref{usami2} (see Figure \ref{transitionsfigure}). 

Let us introduce the family of
excited modes at the frequencies $\sigma^j$ and $\eta_k$,
\begin{align*}
\mathcal{P}_j:=\{\mathcal{E}_{\sigma^j}(\ii B_l), W_\ii(\mathcal{E}_{\sigma^j}(\ii B_l)),\mathcal{E}_{\eta_k}(\ii B_l), W_\ii(\mathcal{E}_{\eta_k}(\ii B_l)) \mid  l=1,2,3, \; k=-j,\dots,j\},
\end{align*}
and notice that, by Lemma~\ref{ext},  
$\mathcal{P}_j \subset \nu_j^1$ (cf.~\eqref{eq:modes}).
Therefore, 
\[
\widetilde{\mathcal{P}}_j:=\{A, [B,C] \mid A,B \in \mathrm{L}_j, C \in \mathcal{P}_j\} \subset \mathcal{T}_j,
\]
where 
we recall that 
$\mathcal{T}_j$ is the minimal ideal of $ \mathrm{Lie}(\nu_j^1)$ containing $\nu_j^0$.

The following proposition, whose proof is given in Appendix~\ref{appendixB}, concludes the proof of Theorem~\ref{rare}.
\begin{proposition}\label{su}
$\mathrm{Lie}(\widetilde{\mathcal{P}}_j)=\mathfrak{su}(\mathcal{M}_j)$.
\end{proposition}

\end{proof}

\begin{remark}
The assumption  $I_2/I_3\notin \mathbb{Q}$  on the moments of inertia appearing in Theorem~\ref{rare} is technical, and prevents the system from having both external resonances (as we saw in Lemma~\ref{ext}) and internal ones (
Lemma~\ref{important}). Anyway, we have not proven that controllability fails if the ratio $I_2/I_3$ is 
rational. 
\end{remark}

\subsection{Reachable sets of the quantum genuine symmetric-top}
In (\ref{alternativa}) we see that, when $\delta = (0,0,\delta_3)^T$, transitions $k \rightarrow k'$ are forbidden if $k \neq k'$. Thus, if the quantum system is prepared in the initial state $\psi(0)$ with $P_3 \psi(0)=k\psi(0)$, the wave function $\psi$ evolves in the subspaces  $S_k=\overline{\mathrm{span}}\{D_{k,m}^j \mid j \in \mathbb{N},m=-j,\dots,j\}$. The next theorem tells us that 
the restriction of (\ref{schro}) to this subspace is approximately controllable.

\begin{theorem}
Let $I_1=I_2$ and fix $k\in \mathbb{Z}$. If $\delta = (0,0,\delta_3)^T$, $\delta_3\neq 0$, then the Schr{\"o}dinger equation (\ref{schro}) is an m-tracker in the Hilbert space $S_k$. In particular, $\mathrm{Reach}(\psi)$ is dense in  $S_k \cap \mathcal{S}$ for all $\psi \in S_k \cap  \mathcal{S}$.
\end{theorem}
\begin{proof}

For every integer $j \ge |k|$, let $I_{j,k}:=\{\rho(l,m)\mid l=j,j+1, \; m=-l,\dots,l\}$, where $\rho: \{(l,m) \mid l\ge |k|,\;m=-l,\dots,l\}\rightarrow \mathbb{N}$ is the lexicographic ordering. 
Then the graph  $\mathcal{G}_k$ with vertices $\{I_{j,k}\}_{j=|k|}^\infty$ and edges
$\{(I_{j,k},I_{j',k})\mid I_{j,k}\cap I_{j',k} \neq \emptyset\}$ is linear.

In order to apply Theorem~\ref{LGTC} to the restriction of  \eqref{schro} to $S_k$, we should consider the projected dynamics onto $\mathcal{N}_{j,k}:= \mathcal{L}_{j,k} \oplus  \mathcal{L}_{j+1,k} $, where $\mathcal{L}_{l,k}:= \mathrm{span}\{D_{k,m}^l\mid m=-l,\dots,l\}$.
The only spectral gaps in $S_k$ are $\sigma^j=|E_k^{j+1}-E_k^j|=\frac{j+1}{I_2}$, $j\ge |k|$. Notice that $(\sigma^j,l)\in \Xi_j^0$.

We write the electric potentials projected onto $\mathcal{N}_{j,k}$:
\begin{align*}
&\mathcal{E}_{\sigma^j}(\ii B_1)=\sum_{m=-j,\dots,j}a_{j,k,m}\delta_3G_{(j,k,m),(j+1,k,m+1)}+a_{j,k,-m}\delta_3G_{(j,k,m),(j+1,k,m-1)}, \\
&\mathcal{E}_{\sigma^j}(\ii B_2)=\sum_{m=-j,\dots,j}a_{j,k,m}\delta_3F_{(j,k,m),(j+1,k,m+1)}-a_{j,k,-m}\delta_3F_{(j,k,m),(j+1,k,m-1)}, \\
&\mathcal{E}_{\sigma^j}(\ii B_3)=\sum_{m=-j,\dots,j}-b_{j,k,m}\delta_3F_{(j,k,m),(j+1,k,m)},
\end{align*}
having used the explicit pairings \eqref{kk}, which can be found in Appendix~\ref{appendixB}, and which describe the transitions excited by the frequency $\sigma^j$.
Note that here the sum does not run over $k$ since we are considering the dynamics restricted to $S_k$.
We consider the family of excited modes 
$$\mathcal{F}_{j,k}=\{\mathcal{E}_{\sigma^j}(\ii B_l), W_\ii(\mathcal{E}_{\sigma^j}(\ii B_l)) \mid l=1,2,3\}\subset \nu_j^0.$$

We claim that the Lie algebra generated by $\mathcal{F}_{j,k}$, seen as a subset of $\mathfrak{su}((2j+1)^2+(2(j+1)+1)^2)$, is equal to $\mathfrak{su}((2j+1)^2+(2(j+1)+1)^2)$.
Such an identity has been proved in \cite[Section 3.3]{BCS}, since the projection
to $\mathcal{N}_{j,k}$ is isomorphic to an analogous projection for the linear molecule. Hence, we conclude that  system \eqref{schro} is an m-tracker in $S_k$. 
\end{proof}

\subsection{Non-controllability of the quantum orthogonal accidentally symmetric-top}
Let us consider separately  the case where $\delta=(\delta_1,\delta_2,0)^T$,  left out by Theorem~\ref{rare}.
 The situation in which the dipole lies in the plane orthogonal to the symmetry axis of the molecule (that is, the \emph{orthogonal} accidentally symmetric-top) is interesting from the point of view of chemistry, since the accidentally symmetric-top molecules present in nature are usually of that kind (see Figure \ref{d2s2}). 
\begin{figure}[ht!]\begin{center}
\includegraphics[width=0.3\linewidth, draft = false]{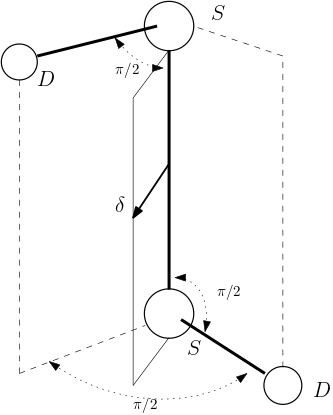}
\caption{Diagram of the orthogonal accidentally symmetric-top approximation of the molecule $D_2S_2$. The electric dipole $\delta$ lies in the orthogonal plane to the symmetry axis.}\end{center} \label{d2s2}
\end{figure}
In order to study this problem, let us introduce the Wang functions
{\cite[Section 7.2]{gordy}}
\[
S_{0,m,0}^j:=D_{0,m}^j,,\qquad 
S_{k,m,\gamma}^j:=\dfrac{1}{\sqrt{2}}(D_{k,m}^j+(-1)^\gamma D_{-k,m}^j), \quad k=1,\dots,j,
\]
for $j\in \mathbb{N}$, $m=-j,\dots,j$, and $\gamma=0,1$. Due to the $k$-degeneracy $E_k^j=E_{-k}^j$ in the spectrum of the rotational Hamiltonian $H$, the functions $S_{k,m,\gamma}^j$ still form an orthogonal basis of eigenfunctions of $H$. Then we consider the change of basis $D_{k,m}^j \rightarrow e^{-\ii k\theta}D_{k,m}^j$, and we choose $\theta \in [0,2\pi)$ such that 
\begin{equation}\label{thetachange}
\begin{cases}
e^{-\ii \theta}(\delta_2+\ii\delta_1)=\ii \sqrt{\delta_1^2+\delta_2^2},  & \\ 
e^{\ii \theta}(\delta_2-\ii\delta_1)=-\ii \sqrt{\delta_1^2+\delta_2^2}. &
\end{cases}
\end{equation}
System \eqref{thetachange} describes the rotation of angle $\mp \theta$ in the complex plane of the vector $\delta_2\pm \ii\delta_1$.
The composition of these two changes of basis gives us the rotated Wang states $S_{k,m,\gamma}^j(\theta):=\frac{1}{\sqrt{2}}(e^{-\ii k\theta}D_{k,m}^j+(-1)^\gamma e^{\ii k\theta} D_{-k,m}^j)$, for $k\neq 0$, and $S_{0,m,0}^j=D_{0,m}^j$.

In the next theorem
we express in this new basis
 a symmetry which prevents the system from being approximately controllable.

\begin{theorem}\label{accidentally}
Let $I_1=I_2$ and $\delta = (\delta_1,\delta_2,0)^T$.
Then the parity of $j+\gamma+k$ is conserved, that is, 
the spaces $\overline{\mathrm{span}}\{S_{k,m,\gamma}^j\mid j+\gamma+k\mbox{ is odd}\}$ and $\overline{\mathrm{span}}\{S_{k,m,\gamma}^j\mid j+\gamma+k\mbox{ is even}\}$
are invariant for the propagators of \eqref{schro}.
\end{theorem}
\begin{proof}
We need to prove that the pairings allowed by the controlled vector fields $B_1,B_2$ and $B_3$ conserve the parity of $j+\gamma+k$. To do so, let us compute
\begin{align}
\langle S_{k,m,\gamma}^j(\theta),\ii B_1 S_{k+1,m+1,\gamma}^{j+1}(\theta)\rangle&=-c_{j,k,m}e^{-\ii \theta}(\delta_2+\ii\delta_1)+c_{j,k,m}e^{\ii \theta}(\delta_2-\ii\delta_1)\nonumber\\ 
&=-2\ii c_{j,k,m}\sqrt{\delta_1^2+\delta_2^2},\label{pairingwang1} \\ 
\langle S_{k,m,\gamma}^j(\theta),\ii B_1 S_{k+1,m+1,\gamma'}^{j+1}(\theta)\rangle&=-c_{j,k,m}e^{-\ii \theta}(\delta_2+\ii\delta_1)-c_{j,k,m}e^{\ii \theta}(\delta_2-\ii\delta_1)\nonumber\\
 &=0, \qquad \gamma\neq \gamma',  \nonumber
\end{align}
having used the expression of the Wang functions as linear combinations of Wigner functions, the explicit pairings (\ref{rotatedwigner}) which can be found in Appendix~\ref{appendixA}, and the choice of $\theta$  made in (\ref{thetachange}). Then we also have 
\begin{equation}\label{pairingwang2}
\begin{cases}
\langle S_{k,m,\gamma}^j(\theta),\ii B_1 S_{k+1,m+1,\gamma}^{j}(\theta)\rangle=0, & \\
\langle S_{k,m,\gamma}^j(\theta),\ii B_1 S_{k+1,m+1,\gamma'}^{j}(\theta)\rangle=-2\ii h_{j,k,m} \sqrt{\delta_1^2+\delta_2^2}, \quad \gamma\neq\gamma', &
\end{cases}
\end{equation}
having used this time the pairings (\ref{jj}), which can be found in Appendix~\ref{appendixB}. From (\ref{pairingwang1}) and (\ref{pairingwang2}) we can see that the allowed transitions only depend on the parity of $j+\gamma$ and $k$; indeed, we have either transitions between states of the form
\[
\begin{cases}
j+\gamma & \text{even} \\
k & \text{even} \\
\end{cases}\longleftrightarrow \begin{cases}
j'+\gamma' & \text{odd} \\
k' & \text{odd}, 
\end{cases}
\]
or transitions between states of the form 
\[
\begin{cases}
j+\gamma & \text{even} \\
k & \text{odd} \\
\end{cases}\longleftrightarrow \begin{cases}
j'+\gamma' & \text{odd} \\
k' & \text{even}. 
\end{cases}
\]
The same happens if we replace $m+1$ with $m-1$ and $k+1$ with $k-1$ in \eqref{pairingwang1} and \eqref{pairingwang2}. Because of the selection rules (\ref{rules}), these are the only transitions allowed by the field $B_1$.
One can easily check, in the same way, that every transition induced by $B_2,B_3$ also conserves the parity of $j+\gamma+k$. \end{proof}

\appendix

\section{Proof of Proposition \ref{propA}}\label{appendixA}As a consequence of Lemma \ref{important}, part 1, if $I_2/I_3\notin \mathbb{Q}$, the only transitions driven by the fields $\ii B_l$, $l=1,2,3$, excited at frequency $\lambda_k^j$, are the ones corresponding to the following matrix elements (written in the basis of $\mathcal{M}_j$ given by the Wigner functions) and can be computed using, e.g., \cite[Table 2.1]{gordy}: 
\begin{equation}\label{kk+1}
\begin{cases}
\langle D_{k,m}^j, \ii B_1 D_{k+1,m\pm1}^{j+1} \rangle=-c_{j,k,\pm m}(\delta_2+\ii\delta_1),  & \\
\langle D_{k,m}^j , \ii B_1 D_{k-1,m\pm1}^{j+1} \rangle=c_{j,-k,\pm m}(\delta_2-\ii\delta_1), & \\
\langle D_{k,m}^j , \ii B_2 D_{k+1,m\pm1}^{j+1} \rangle= \mp \ii c_{j,k,\pm m} (\delta_2+\ii\delta_1),  & \\
\langle D_{k,m}^j , \ii B_2 D_{k-1,m\pm1}^{j+1} \rangle=\pm \ii c_{j,-k,\pm m}(\delta_2-\ii\delta_1), & \\
\langle D_{k,m}^j , \ii B_3 D_{k\pm1,m}^{j+1} \rangle=\pm \ii d_{j,\pm k,m}(\delta_2\pm\ii\delta_1), & 
\end{cases}
\end{equation}
where 
$$ c_{j,k,m}:=  \dfrac{ [(j+k+1)(j+k+2)]^{1/2}[(j+m+1)(j+m+2)]^{1/2}}{4(j+1)[(2j+1)(2j+3)]^{1/2}},$$
and
$$ d_{j,k,m}:=  \dfrac{ [(j+k+1)(j+k+2)]^{1/2} [(j+1)^2-m^2]^{1/2}}{2(j+1)[(2j+1)(2j+3)]^{1/2}}.$$

Now, using a symmetry argument, we explain how to get rid of one electric dipole component between $\delta_1$ and $\delta_2$.

By the very definition of the Euler angles, one has that the rotation of angle $\theta$ around the symmetry axis $a_3$ is given by $\gamma \mapsto \gamma+\theta.$ This rotation acts on the Wigner functions in the following way 
$$D_{k,m}^j(\alpha,\beta,\gamma)\mapsto D_{k,m}^j(\alpha,\beta,\gamma+\theta)=e^{\ii k\theta}D_{k,m}^j(\alpha,\beta,\gamma)=:D_{k,m}^j(\theta)(\alpha,\beta,\gamma),$$
having used the explicit expression of the symmetric states (\ref{explicit}). Note that these rotated Wigner functions form again an orthogonal basis for $L^2({\rm SO}(3))$ of eigenfunctions of the rotational Hamiltonian $H$, so we can also analyze the controllability problem in this new basis. 
In this new basis the matrix elements (corresponding to the frequency $\lambda_k^j$) of the controlled fields are
\begin{equation}\label{rotatedwigner}
\begin{cases}
\langle D_{k,m}^j(\theta) , \ii B_1 D_{k+1,m+1}^{j+1}(\theta) \rangle&= -c_{j,k,m}e^{-\ii \theta}(\delta_2+\ii\delta_1),\\
\langle D_{k,m}^j(\theta) , \ii B_1 D_{k-1,m+1}^{j+1}(\theta) \rangle&= c_{j,-k,m} e^{\ii \theta} (\delta_2-\ii\delta_1),
\end{cases}
\end{equation}
and the same happens for all the other transitions described in (\ref{kk+1}). So, the effect of this change of basis is that we are actually rotating the first two components of the dipole moment, by the angle $\theta$. We can now choose $\theta\in [0,2\pi)$ such that 
 $$e^{-\ii \theta}(\delta_2+\ii\delta_1)=\sqrt{\delta_1^2+\delta_2^2}=e^{\ii \theta} (\delta_2-\ii\delta_1). $$
In other words, thanks to this change of basis, we can assume without loss of generality that $\delta_1=0$, since we can rotate the vector $\delta_2\pm \ii\delta_1$ and get rid of its imaginary part (note that in \eqref{thetachange} and in the proof of Theorem \ref{accidentally} we are rotating the vector $\delta_2\pm \ii\delta_1$ in the other sense, i.e., to get rid of its real part). This will simplify the expression of the controlled fields. Note that 
\begin{align*}
&W_\ii(G_{(j,k,m),(j+1,k+1,n)})=-F_{(j,k,m),(j+1,k+1,n)}, \\ & W_\ii(F_{(j,k,m),(j+1,k+1,n)})=G_{(j,k,m),(j+1,k+1,n)}.
\end{align*}
From the identity $[e_{j,k},e_{n,m}]=\delta_{kn}e_{j,m}-\delta_{jm}e_{n,k}$ we get the 
bracket relations
 $$[G_{j,k},G_{k,n}]=G_{j,n}, \quad [F_{j,k},F_{k,n}]=-G_{j,n}, \quad [G_{j,k},F_{k,n}]=F_{j,n},$$
$$ [G_{j,k},F_{j,k}]=2D_{j,k}, \quad [F_{j,k},D_{j,k}]=2G_{j,k}.$$
Moreover, 
two operators coupling no common states commute, that is,
$$ [Y_{j,k},Z_{j',k'}]=0 \qquad \text{ if } \{j,k\}\cap  \{j',k'\}=\emptyset ,$$
with $Y,Z\in \{G,F,D\}$.

Finally, 
we can conveniently represent the matrices corresponding to the controlled vector field (projected onto $\mathcal{M}_j$) in the rotated basis found with the symmetry argument. So, for each $k=-j,\dots,j$, because of Lemma \ref{important}, part 1, and \eqref{kk+1}, we have
\begin{align}\nonumber
\mathcal{E}_{\lambda_k^j}(\ii B_1)=&\sum_{m=-j,\dots,j}  -c_{j,k,m}\delta_2 G_{(j,k,m),(j+1,k+1,m+1)}-c_{j,k,-m}\delta_2 G_{(j,k,m),(j+1,k+1,m-1)} \\ 
\label{changeofbasis1}
& +c_{j,k,m}\delta_2 G_{(j,-k,m),(j+1,-k-1,m+1)}+c_{j,k,-m}\delta_2 G_{(j,-k,m),(j+1,-k-1,m-1)},
\end{align}
\begin{align}\nonumber
\mathcal{E}_{\lambda_k^j}(\ii B_2)=&\sum_{m=-j,\dots,j}-c_{j,k,m}\delta_2F_{(j,k,m),(j+1,k+1,m+1)}+c_{j,k,-m}\delta_2F_{(j,k,m),(j+1,k+1,m-1)} \\ 
\label{changeofbasis2}
&+c_{j,k,m}\delta_2F_{(j,-k,m),(j+1,-k-1,m+1)}-c_{j,k,-m}\delta_2F_{(j,-k,m),(j+1,-k-1,m-1)},
\end{align}
\begin{align}\label{changeofbasis3}
\mathcal{E}_{\lambda_k^j}(\ii B_3)&=\sum_{m=-j,\dots,j}d_{j,k,m}\delta_2F_{(j,k,m),(j+1,k+1,m)} -d_{j,k,m}\delta_2F_{(j,-k,m),(j+1,-k-1,m)},
\end{align}
where, with a slight abuse of notation, we write $B_l$ instead of $\Pi_{\mathcal{M}_j}B_l\Pi_{\mathcal{M}_j}$.

Now we show how the sum over $m$ in (\ref{changeofbasis1}), (\ref{changeofbasis2}) and (\ref{changeofbasis3}) can be decomposed, in order to obtain that 
the 
matrices $X_{(j,k,m),(j+1,k+1,m\pm1)}+X_{(j,-k,m),(j+1,-k-1,m\pm1)}$ and $X_{(j,k,m),(j+1,k+1,m)}-X_{(j,-k,m),(j+1,-k-1,m)}$ are in $\mathrm{L}_j$, for any $m,k=-j,\dots,j$, where $X\in\{G,F\}$.
Let us first fix $k\neq 0$ and consider 
 \begin{align*}
 W_\ii(\mathcal{E}_{\lambda_k^j}&(\ii B_3))\\
 &=\sum_{m=-j,\dots,j}d_{j,k,m}\delta_2G_{(j,k,m),(j+1,k+1,m)} -d_{j,k,m}\delta_2G_{(j,-k,m),(j+1,-k-1,m)},
 \end{align*}
  and the brackets 
\begin{align*}
\mathrm{ad}_{\mathcal{E}_{\lambda_k^j}(\ii B_3)}^{2s}(W_\ii(\mathcal{E}_{\lambda_k^j}(\ii B_3)))=&\sum_{m=-j,\dots,j}(-1)^s2^{2s}d_{j,k,m}^{2s+1}\delta_2^{2s+1}G_{(j,k,m),(j+1,k+1,m)}\\ &+(-1)^s2^{2s}(-d_{j,k,m})^{2s+1}\delta_2^{2s+1}G_{(j,-k,m),(j+1,-k-1,m)},
\end{align*}
for $s\in \mathbb{N}$, where $\mathrm{ad}_A(B)=[A,B]$ and $\mathrm{ad}_A^{n+1}(B)=[A,\mathrm{ad}_A^n(B)]$.  Since $d_{j,k,m}=d_{j,k,-m}$, the invertibility of the Vandermonde matrix gives that 
\begin{align}\nonumber
&G_{(j,k,m),(j+1,k+1,m)}+G_{(j,k,-m),(j+1,k+1,-m)} \\ -&G_{(j,-k,m),(j+1,-k-1,m)}-G_{(j,-k,-m),(j+1,-k-1,-m)} \in \mathrm{L}_j,\label{vandermonde}
\end{align}
for $m=0,\dots,j$. In particular, $G_{(j,k,0),(j+1,k+1,0)}-G_{(j,-k,0),(j+1,-k-1,0)}$ is in $\mathrm{L}_j$. Hence, 
\begin{align}
\label{replace}& \Big[\Big[\dfrac{\mathcal{E}_{\lambda_k}(\ii B_1)-W_\ii(\mathcal{E}_{\lambda_k}(\ii B_2))}{2},G_{(j,k,0),(j+1,k+1,0)}-G_{(j,-k,0),(j+1,-k-1,0)}\Big],\\ 
\nonumber
&G_{(j,k,0),(j+1,k+1,0)}-G_{(j,-k,0),(j+1,-k-1,0)}\Big]=c_{j,k,0}\delta_2G_{(j,k,0),(j+1,k+1,-1)}\\ & +\nonumber c_{j,k,-1}\delta_2G_{(j,k,1),(j+1,k+1,0)}-c_{j,k,0}\delta_2G_{(j,-k,0),(j+1,-k-1,-1)}\\ &-c_{j,k,-1}\delta_2G_{(j,-k,1),(j+1,-k-1,0)} 
\nonumber 
\end{align}
is also in $\mathrm{L}_j$. Define 
\begin{align*}
Q_0=&c_{j,k,0}\delta_2G_{(j,k,0),(j+1,k+1,-1)}+c_{j,k,-1}\delta_2G_{(j,k,1),(j+1,k+1,0)}\\
&-c_{j,k,0}\delta_2G_{(j,-k,0),(j+1,-k-1,-1)}-c_{j,k,-1}\delta_2G_{(j,-k,1),(j+1,-k-1,0)}, \\
Q_m=&c_{j,k,-m}\delta_2G_{(j,k,-m),(j+1,k+1,-m-1)}+c_{j,k,-m-1}\delta_2G_{(j,k,m+1),(j+1,k+1,m)}\\
&-c_{j,k,-m}\delta_2G_{(j,-k,-m),(j+1,-k-1,-m-1)}-c_{j,k,-m-1}\delta_2G_{(j,-k,m+1),(j+1,-k-1,m)} ,
\end{align*}
if $ 0<m<j$, and 
\[Q_j=c_{j,k,-j}\delta_2G_{(j,k,-j),(j+1,k+1,-j-1)}-c_{j,k,-j}\delta_2G_{(j,-k,-j),(j+1,-k-1,-j-1)}.\]

We have 
\begin{align*}
&\Big[\Big[\sum_{m=s,\dots,j} Q_{m},G_{(j,k,s),(j+1,k+1,s)}+G_{(j,k,-s),(j+1,k+1,-s)}-G_{(j,-k,s),(j+1,-k-1,s)}\\ &-G_{(j,-k,-s),(j+1,-k-1,-s)}\Big] 
,G_{(j,k,s),(j+1,k+1,s)}+G_{(j,k,-s),(j+1,k+1,-s)}\\ & -G_{(j,-k,s),(j+1,-k-1,s)}-G_{(j,-k,-s),(j+1,-k-1,-s)}\Big]=Q_s,
\end{align*}
for $s=1,\dots,j$. By iteration on $s$ and because of (\ref{vandermonde}), it follows that $Q_s\in \mathrm{L}_j$ for every $s=0,\dots,j$. Now, since 
\[
\dfrac{Q_j}{c_{j,k,-j}\delta_2}=G_{(j,k,-j),(j+1,k+1,-j-1)}-G_{(j,-k,-j),(j+1,-k-1,-j-1)}\in \mathrm{L}_j,\] then
\begin{align*}
&\mathrm{ad}_{G_{(j,k,-j),(j+1,k+1,-j-1)}-G_{(j,-k,-j),(j+1,-k-1,-j-1)}}^2(G_{(j,k,j),(j+1,k+1,j)}\\ &+G_{(j,k,-j),(j+1,k+1,-j)} 
-G_{(j,-k,j),(j+1,-k-1,j)}-G_{(j,-k,-j),(j+1,-k-1,-j)})\\ & =G_{(j,k,-j),(j+1,k+1,-j)}-G_{(j,-k,-j),(j+1,-k-1,-j)}\in \mathrm{L}_j,
\end{align*} 
which, in turns, implies that 
\begin{align*}
&\mathrm{ad}_{G_{(j,k,-j),(j+1,k+1,-j)}-G_{(j,-k,-j),(j+1,-k-1,-j)}}^2(Q_{j-1}) \\ & = c_{j,k,-j+1}G_{(j,k,-j+1),(j+1,k+1,-j)}-c_{j,k,-j+1}G_{(j,-k,-j+1),(j+1,-k-1,-j)}\in \mathrm{L}_j.
\end{align*}

Iterating the argument, 
\begin{equation}\label{example}
G_{(j,k,m),(j+1,k+1,m)}-G_{(j,-k,m),(j+1,-k-1,m)} \in \mathrm{L}_j, \quad m=-j,\dots,j
\end{equation}
and $G_{(j,k,m),(j+1,k+1,m-1)}-G_{(j,-k,m),(j+1,-k-1,m-1)}$ are in $\mathrm{L}_j$ for $m=-j,\dots,j$.

By the same argument as above, with 
$\frac{\mathcal{E}_{\lambda_k^j}(\ii B_1)-W_\ii(\mathcal{E}_{\lambda_k^j}(\ii B_2))}{2}$
replaced by 
\begin{align*}
\dfrac{\mathcal{E}_{\lambda_k^j}(\ii B_1)+W_\ii(\mathcal{E}_{\lambda_k^j}(\ii B_2))}{2}=&\sum_{m=-j,\dots,j}-c_{j,k,m}G_{(j,k,m),(j+1,k+1,m+1)}\\ &+c_{j,k,m}G_{(j,-k,m),(j+1,-k-1,m+1)}
\end{align*}
in 
(\ref{replace}) 
we also have that $G_{(j,k,m),(j+1,k+1,m+1)}-G_{(j,-k,m),(j+1,-k-1,m+1)}$ is in $\mathrm{L}_j$ for all $m=-j,\dots,j$.

 If we now replace 
$ \frac{\mathcal{E}_{\lambda_k^j}(\ii B_1)-W_\ii(\mathcal{E}_{\lambda_k^j}(\ii B_2))}{2}$
with 
\begin{align*}
\dfrac{\mathcal{E}_{\lambda_k^j}(\ii B_2)+W_\ii(\mathcal{E}_{\lambda_k^j}(\ii B_1))}{2}=&\sum_{m=-j,\dots,j}-c_{j,k,-m}F_{(j,k,m),(j+1,k+1,m-1)}\\ &+c_{j,k,-m}F_{(j,-k,m),(j+1,-k-1,m-1)}
\end{align*}
 or 
\begin{align*}
\dfrac{\mathcal{E}_{\lambda_k^j}(\ii B_2)-W_\ii(\mathcal{E}_{\lambda_k^j}(\ii B_1))}{2}=&\sum_{m=-j,\dots,j}-c_{j,k,m}F_{(j,k,m),(j+1,k+1,m+1)}\\ &+c_{j,k,m}F_{(j,-k,m),(j+1,-k-1,m+1)},
\end{align*}
the arguments above prove that  
both $F_{(j,k,m),(j+1,k+1,m)}-F_{(j,-k,m),(j+1,-k-1,m)}$ and
$ 
F_{(j,k,m),(j+1,k+1,m\pm1)}-F_{(j,-k,m),(j+1,-k-1,m\pm1)}$
 are in $\mathrm{L}_j$ for all $m=-j,\dots,j$.

\section{Proof of Proposition \ref{su}}\label{appendixB}
Using again \cite[Table 2.1]{gordy} we write the pairings
\begin{equation}\label{jj}
\begin{cases}
\langle D_{k,m}^j , \ii B_1D_{k+1,m\pm1}^{j} \rangle=\mp h_{j,k,\pm m}(\delta_2+\ii \delta_1), & \\
\langle D_{k,m}^j , \ii B_1D_{k-1,m\pm 1}^{j} \rangle=\mp h_{j,-k,\pm m}(\delta_2-\ii \delta_1), & \\
\langle D_{k,m}^j , \ii B_2D_{k+1,m\pm1}^{j} \rangle= -\ii h_{j,k,\pm m}(\delta_2+\ii \delta_1), & \\
\langle D_{k,m}^j , \ii B_2D_{k-1,m\pm1}^{j} \rangle=-\ii h_{j,-k,\pm m}(\delta_2- \ii \delta_1),  & \\
\langle D_{k,m}^j , \ii B_3D_{k\pm1,m}^{j} \rangle=-\ii q_{j,\pm k,m}(\delta_2 \pm \ii \delta_1), & 
\end{cases}
\end{equation}
where 
\begin{align*}
 h_{j,k,m}&:=  \dfrac{[j(j+1)-k(k+1)]^{1/2}[j(j+1)-m(m+1)]^{1/2}}{4j(j+1)}, \\
q_{j,k,m}&:=  \dfrac{[j(j+1)-k(k+1)]^{1/2}m}{2j(j+1)}.
\end{align*}
Moreover,
\begin{equation}\label{kk}
\begin{cases}
\langle D_{k,m}^j , \ii B_1D_{k,m\pm 1}^{j+1} \rangle=a_{j,k,\pm m}\delta_3, & \\
\langle D_{k,m}^j , \ii B_2D_{k,m\pm 1}^{j+1}\rangle=\pm \ii a_{j,k,\pm m}\delta_3, & \\
\langle D_{k,m}^j , \ii B_3D_{k,m}^{j+1} \rangle=-\ii b_{j,k,m}\delta_3, & 
\end{cases}
\end{equation}
where
\begin{align*}
a_{j,k,m}&:= \dfrac{[(j+1)^2-k^2]^{1/2}[(j+m+1)(j+m+2)]^{1/2}}{2(j+1)[(2j+1)(2j+3)]^{1/2}},\\
b_{j,k,m}&:=  \dfrac{[(j+1)^2-k^2]^{1/2}[(j+1)^2-m^2]^{1/2}}{(j+1)[(2j+1)(2j+3)]^{1/2}}.
\end{align*}

Note that the $k\rightarrow k$ transitions are driven by $\delta_3$. 
Recall that, up to a rotation, we can 
assume that $\delta_1=0$. Because of Lemma \ref{important}, parts 2 and 3, the expression of the controlled fields excited at the frequencies $\eta_k$ and $\sigma^j$ are
\begin{align}\nonumber
\mathcal{E}_{\eta_k}(\ii B_1)=&\sum_{\substack{l=j,j+1,\\ m=-l,\dots,l-1}}-h_{l,k,m}\delta_2G_{(l,k,m),(l,k+1,m+1)} -h_{l,k,m}\delta_2G_{(l,-k,m),(l,-k-1,m+1)} \\ 
\label{eta2}& 
\hspace{-1mm}+\hspace{-2mm}\sum_{\substack{l=j,j+1,\\ m=-l+1,\dots,l}} h_{l,k,-m}\delta_2G_{(l,k,m),(l,k+1,m-1)}+h_{l,k,-m}\delta_2G_{(l,-k,m),(l,-k-1,m-1)} ,
\end{align}
\begin{align}\nonumber
\mathcal{E}_{\eta_k}(\ii B_2)&=\sum_{\substack{l=j,j+1,\\m=-l,\dots,l-1}}-h_{l,k,m}\delta_2F_{(l,k,m),(l,k+1,m+1)} -h_{l,k,m}\delta_2F_{(l,-k,m),(l,-k-1,m+1)}  \\ & 
+\hspace{-1mm}\sum_{\substack{l=j,j+1,\\m=-l+1,\dots,l}}-h_{l,k,-m}\delta_2F_{(l,k,m),(l,k+1,m-1)}-h_{l,k,-m}\delta_2F_{(l,-k,m),(l,-k-1,m-1)} ,
\label{eta3}
\end{align}
\begin{equation}\label{eta4}
\mathcal{E}_{\eta_k}(\ii B_3)=\sum_{\substack{l=j,j+1,\\m=-l,\dots,l}}-q_{l,k,m}\delta_2F_{(l,k,m),(l,k+1,m)} -q_{l,k,m}\delta_2F_{(l,-k,m),(l,-k-1,m)},
\end{equation}
and
\begin{align}\label{first}
&\mathcal{E}_{\sigma^j}(\ii B_1)=\sum_{m,k=-j,\dots,j}a_{j,k,m}\delta_3G_{(j,k,m),(j+1,k,m+1)}+a_{j,k,-m}\delta_3G_{(j,k,m),(j+1,k,m-1)}, \\
&\mathcal{E}_{\sigma^j}(\ii B_2)=\sum_{m,k=-j,\dots,j}a_{j,k,m}\delta_3F_{(j,k,m),(j+1,k,m+1)}-a_{j,k,-m}\delta_3F_{(j,k,m),(j+1,k,m-1)},\nonumber \\
&\mathcal{E}_{\sigma^j}(\ii B_3)=\sum_{m,k=-j,\dots,j}-b_{j,k,m}\delta_3F_{(j,k,m),(j+1,k,m)} .\nonumber
\end{align}
Note that in $\mathcal{E}_{\eta_k}(-iB_3)$ the term for $m=0$ vanishes, since $q_{j,k,0}=0$ for every $j,k$.

To decouple all the $m$-degeneracies in the excited modes, we just consider double brackets with the elementary matrices that we have obtained above. As an example, using \eqref{example} we can decouple the $m\rightarrow m$ transitions corresponding to the frequency $\sigma^j$ by considering 
\begin{align*}
&[[W_\ii(\mathcal{E}_{\sigma^j}(\ii B_3)),G_{(j,k,m),(j+1,k+1,m)}-G_{(j,-k,m),(j+1,-k-1,m)}],  \\ &
G_{(j,k,m),(j+1,k+1,m)}-G_{(j,-k,m),(j+1,-k-1,m)}]\\ &=G_{(j,k,m),(j+1,k,m)}+G_{(j,-k,m),(j+1,-k,m)} \in \mathrm{Lie}(\widetilde{\mathcal{P}}_j).
\end{align*}
Considering every possible double brackets as above, we obtain, for $X \in \{G,F \}$, that
\begin{equation}\label{example2}
X_{(j,k,m),(j+1,k,m)}+X_{(j,-k,m),(j+1,-k,m)}\in \mathrm{Lie}(\widetilde{\mathcal{P}}_j), \quad k\neq 0,
\end{equation}
when we start from the matrices in \eqref{first}, and that
\[X_{(l,k,m),(l,k+1,m)}+X_{(l,-k,m),(l,-k-1,m)},X_{(l,k,m),(l,k+1,m\pm1)}+X_{(l,-k,m),(j,-k-1,m\pm1)}\]
are in $\mathrm{Lie}(\widetilde{\mathcal{P}}_j)$,  $l=j,j+1$, $m,k\neq 0$,  when we start from the matrices in \eqref{eta2},  \eqref{eta3}, \eqref{eta4}. 
Now we can also generate the missing $k=0$ elements of \eqref{example} by taking double brackets with $X_{(j+1,1,m),(j+1,2,m)}+X_{(j+1,-1,m),(j+1,-2,m)}\in \mathrm{Lie}(\widetilde{\mathcal{P}}_j)$. As an example, we have that
\begin{align*}
&[[\mathcal{E}_{\lambda_0^j}(\ii B_3),F_{(j+1,1,m),(j+1,2,m)}+F_{(j+1,-1,m),(j+1,-2,m)}],  \\ &
F_{(j+1,1,m),(j+1,2,m)}+F_{(j+1,-1,m),(j+1,-2,m)}]\\ &=F_{(j,0,m),(j+1,1,m)}-F_{(j,0,m),(j+1,-1,m)}\in \mathrm{Lie}(\widetilde{\mathcal{P}}_j). 
\end{align*}
Moreover, also the $m=0$ elements in the transitions \eqref{eta4} are in $\mathrm{Lie}(\widetilde{\mathcal{P}}_j)$, as one can check by considering a bracket between two transitions obtained in \eqref{example} and \eqref{example2}. For example, 
\begin{align*}
&[G_{(j,k,0),(j+1,k+1,0)}-G_{(j,-k,0),(j+1,-k-1,0)},G_{(j+1,k+1,0),(j,k+1,0)}\\ &+G_{(j+1,-k-1,0),(j,-k-1,0)}]  
= G_{(j,k,0),(j,k+1,0)}-G_{(j,-k,0),(j,-k-1,0)}\in \mathrm{Lie}(\widetilde{\mathcal{P}}_j).
\end{align*}
Finally, we apply a three-wave mixing argument (Figure \ref{wave}) in order to decouple the sum over $k$ and $-k$ in every elementary matrices: consider the bracket between the following elements in $\mathrm{Lie}(\widetilde{\mathcal{P}}_j)$
\begin{figure}[ht!]\begin{center}
\includegraphics[width=0.6\linewidth, draft = false]{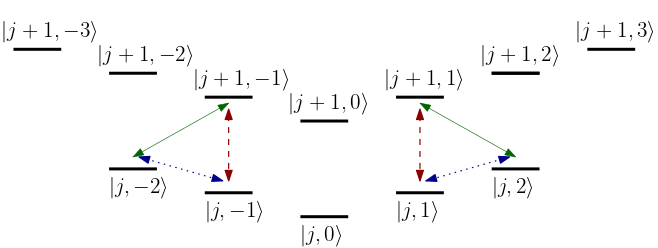}
\caption{Three-wave mixing around $k=1,-1$. The same-shaped arrows correspond to equal spectral gaps, and thus, coupled transitions. The goal of the three-wave mixing is to decouple those arrows.}\end{center} \label{wave}
\end{figure}
\begin{align*}
&[G_{(j,k+1,m),(j,k,m)}+G_{(j,-k-1,m),(j,-k,m)},G_{(j,k,m),(j+1,k,m)}+G_{(j,-k,m),(j+1,-k,m)}]  \\ &
= G_{(j,k+1,m),(j+1,k,m)}+G_{(j,-k-1,m),(j+1,-k,m)} \in \mathrm{Lie}(\widetilde{\mathcal{P}}_j), \quad k\neq0,
\end{align*}
and notice that from \eqref{example} we already have that 
$G_{(j,k+1,m),(j+1,k,m)}-G_{(j,-k-1,m),(j+1,-k,m)}$
is in $\mathrm{Lie}(\widetilde{\mathcal{P}}_j)$, and hence $G_{(j,k+1,m),(j+1,k,m)}$ and $G_{(j,-k-1,m),(j+1,-k,m)}$ are in $\mathrm{Lie}(\widetilde{\mathcal{P}}_j)$. In this way we can break every $k$-degeneracy, and finally obtain that $\mathrm{Lie}(\widetilde{\mathcal{P}}_j)=\mathfrak{su}(\mathcal{M}_j)$, which concludes the proof.

\bibliographystyle{siamplain}
\bibliography{references}

\end{document}